\documentclass[conference]{IEEEtran}
\IEEEoverridecommandlockouts
\usepackage[mathscr]{euscript}
\usepackage[english]{babel}
\usepackage{cite}
\usepackage{amsthm}
\usepackage{amsmath,amssymb,amsfonts}
\usepackage{subcaption}
\usepackage{booktabs}
\usepackage{bigstrut}
\usepackage{multirow}
\newcommand{\tabitem}{~\llap{\textbullet}~}
\usepackage{algorithmic}
\usepackage{graphicx}
\usepackage{textcomp}
\usepackage{xcolor}
\usepackage{algorithm,algorithmic}
\newtheorem{problem}{Problem}[section]
\newtheorem{proposition}{Proposition}[section]

\def\BibTeX{{\rm B\kern-.05em{\sc i\kern-.025em b}\kern-.08em
    T\kern-.1667em\lower.7ex\hbox{E}\kern-.125emX}}

\begin{document}

\title{Creating Realistic Power Distribution Networks using Interdependent Road Infrastructure}

\author{\IEEEauthorblockN{Rounak Meyur, Madhav Marathe, Anil Vullikanti,\\Henning Mortveit, Samarth Swarup}
\IEEEauthorblockA{\textit{Biocomplexity Institute, University of Virginia} \\
Charlottesville, Virginia \\
\emph{\{rm5nz, marathe, vsakumar, henning.mortveit, swarup\}@virginia.edu}
}
\and
\IEEEauthorblockN{Virgilio Centeno, Arun Phadke}
\IEEEauthorblockA{\textit{ECE Department, Virginia Tech} \\
Blacksburg, Virginia \\
\emph{\{virgilio, aphadke\}@vt.edu}}
}

\maketitle
\thispagestyle{plain}
\pagestyle{plain}

\begin{abstract}
It is well known that physical interdependencies exist between networked civil infrastructures such as transportation and power system networks. In order to analyze complex nonlinear correlations between such networks, datasets pertaining to such real infrastructures are required. However, such data are not readily available due to their proprietary nature. This work proposes a methodology to generate realistic synthetic power distribution networks for a given geographical region. A network generated in this manner is not the actual distribution system, but its functionality is very similar to the real distribution network. The synthetic network connects high voltage substations to individual residential consumers through primary and secondary distribution networks. Here, the distribution network is generated by solving an optimization problem which minimizes the overall length of the network subject to structural and power flow constraints. This work also incorporates identification of long high voltage feeders originating from substations and connecting remotely situated customers in rural geographic locations while maintaining voltage regulation within acceptable limits. The proposed methodology is applied to the state of Virginia and creates synthetic distribution networks which are validated by comparing them to actual power distribution networks at the same location.
\end{abstract}

\begin{IEEEkeywords}
synthetic distribution networks, radial networks, Mixed Integer Linear Programming
\end{IEEEkeywords}

\section{Introduction}
Networked infrastructures are vital components of society and instrumental for its proper functioning~\cite{netinfr}. A common structural attribute of such networked infrastructures is that they are located close to human settlements. Examples of some of these networked infrastructures include power system network, communication network, transportation network etc. Furthermore, the interdependencies among these networked infrastructures are well known~\cite{Adiga_2020}. 

The US government has termed these networks as critical infrastructures since they provide \emph{enabling functions} across other infrastructure sectors~\cite{ppd21}. A failure in one network can lead to cascading events in other dependent networked infrastructures and eventually cause an immense impact on the nation's economy. In recent times, there has been an increased interest in studying the resilience and reliability of such networks during extreme events. The complex non-linear correlations between different networks have drawn the attention of many researchers~\cite{Adiga_2020}. However, a fundamental barrier to progress in this area is the lack of data pertaining to such networked infrastructures. 

On another front, emerging technologies and policies are being proposed to make these networks resilient to failures. For example, in the case of a power distribution system, the role of distributed energy resources (DERs) is being widely acknowledged as an efficient means to enhance the network resiliency. The focus of power systems research has moved towards efficient deployment and control of DERs in the distribution system~\cite{manish2019,manish2018OptimalDS,rad_prot,rad_der,radiality_2012}. However, there remains a lack of realistic test case scenarios to implement the proposed methodologies for their validation. Most of the algorithms and technologies are validated against standard IEEE test feeders or small-sized distribution networks (e.g. 13, 37, 123 bus feeders etc.~\cite{feeders}). The unavailability of appropriate network data is primarily due to the proprietary nature of the data as well as its enormous volume~\cite{review2017}.

\noindent\textbf{Problem} 
The abstract problem we study can be formally stated as: \emph{Given a set of residences and electric substations with their respective geographical coordinates, construct a power distribution network connecting these points that resembles a real operating distribution system network.}
This work proposes a methodology to generate realistic synthetic distribution networks for a geographical region using open source information. The generated network connects substations to individual residential buildings through medium voltage (MV) and low-voltage (LV) networks. The synthetic networks are not exact replicas of the actual distribution network; rather they resemble the actual network in terms of structural properties. At the same time, the generated network satisfies the operational constraints (voltage and power flows limits) of an actual operating distribution system. 

\begin{table*}[htb]
	\centering
	\caption{Major contributions of proposed work}
	\begin{tabular}{p{5em}p{25em}p{25em}}
		\toprule
		\textbf{Aspect} & \textbf{Previous Works} & \textbf{Present work} 
		\\\midrule
		Network type & Synthetic transmission networks with generators and aggregated loads~\cite{overbye_101,overbye_102,trpovski_2018,gm2016,rnm_2011} & Synthetic distribution networks are created which connect high voltage substations to individual residential consumers.
		\\\midrule
		\multirow{3}{5em}{Realism of generated networks} 
		& \multicolumn{1}{p{25em}}{Statistical distribution of network attributes are used to generate synthetic power grids~\cite{schweitzer}}
		& \multirow{3}{25em}{Realistic distribution networks comprising of primary and secondary levels are generated for a given geographical location. The generated network resembles an optimal network designed by power distribution companies. It follows the usual structural and power flow constraints of a typical distribution system.}
		\\\cmidrule{2-2}
		& \multicolumn{1}{p{25em}}{Stochastic geometry based approach to place transformers in distribution networks~\cite{overbye_2019}}&\\\cmidrule{2-2}
		& \multicolumn{1}{p{25em}}{Heuristic approach to synthesize distribution networks from substations and populate them with consumer loads~\cite{trpovski_2018}}
		&
		\\\midrule
		Radiality of network 
		& Radiality is ensured by avoiding isolated cycles or considering single commodity flow model~\cite{radiality_1987,radiality_2012,lei2019radiality,wang2019radial}. 
		& Power balance constraint is proved to be a sufficient condition to ensure radiality of the generated LV network. 
		\\\midrule
		\multirow{2}{5em}{Network attributes} 
		& \multicolumn{1}{p{25em}}{Small sized networks such as standard IEEE test systems are used~\cite{manish2019,manish2018OptimalDS}.}
		& \multirow{2}{25em}{An optimal radial network is identified for an unknown number of root nodes (high voltage feeders) and for large-sized networks with more than 20000 nodes.} \\\cmidrule{2-2}
		& \multicolumn{1}{p{25em}}{The number of root nodes (substations) are known beforehand in the problem definition.} 
		&  
		\\\bottomrule
	\end{tabular}
	\label{tab:contrib}
\end{table*}
\noindent\textbf{Contribution}~Our main contributions are: (i) We develop a \emph{first principles}-based methodology to create a realistic synthetic distribution network using information from other infrastructures such as transportation networks, residential data, etc.; (ii) Our approach results in an optimal network by minimizing the overall length of distribution lines which is a principal consideration of power companies while planning distribution networks; (iii) Our method generates a distribution network which is particular to the geographical location of interest and hence provides a realistic representation of the actual network; (iv) The nodes and edges of the generated network are labeled with all necessary attributes required for power flow analysis and therefore, can be used as suitable networks to test distribution system planning and operation algorithms.

Our approach combines real world data such as building information and road network data with general power engineering practices to construct realistic distribution systems. Instead of considering aggregated loads at zipcode centers, our method generates a detailed distribution network for a given geographical area. The economic aspect of minimizing the total length of the network (which is indicative of the investment required for building a network) has been considered as the primary objective. We include constraints motivated by power engineering practices such as maintaining a tree structure (for protection coordination), placing local transformers (pole-top and pad-mounted) along the road network and avoiding branching in the secondary network (to reduce voltage sag at the leaf nodes). We also use detailed synthetic load demand profiles from a recent model in the literature~\cite{swapna_2018}. The inclusion of all these aspects has allowed us to validate power flow characteristics for the created synthetic networks.

The geographical attributes of the region are retained intact in our methodology. For example, in a rural region with multiple remote localities, a single substation feeder serving all of them would often suffer from voltage sag issues at the extreme ends. In this paper, we consider several substation feeders serving such remote localities. In this way, we ensure that the node voltages and edge flows are within acceptable limits while the radial structure of the network is maintained. The key differences of our method with other related previous methods to generate synthetic distribution networks are listed in Table~\ref{tab:contrib}.

\section{Preliminaries}\label{sec:prelim}
\subsection{Distribution system}\label{ssec:dist}
The distribution network consisting of overhead power lines, underground cables, pole top transformers is responsible to bring electrical power from high voltage (HV; greater than 33kV) transmission system to the end residential consumers requiring a LV level (208-480V). This is normally done in a two-step procedure: (i) the high transmission level voltage is stepped down to MV level at distribution substations and distributed to local transformers (pole-top/pad-mounted) through \emph{primary distribution network}, (ii) the voltage is further stepped down to LV at the local distribution transformers and distributed to individual customers through \emph{secondary distribution network}. 

\begin{figure}
	\centering
	\includegraphics[width=0.235\textwidth]{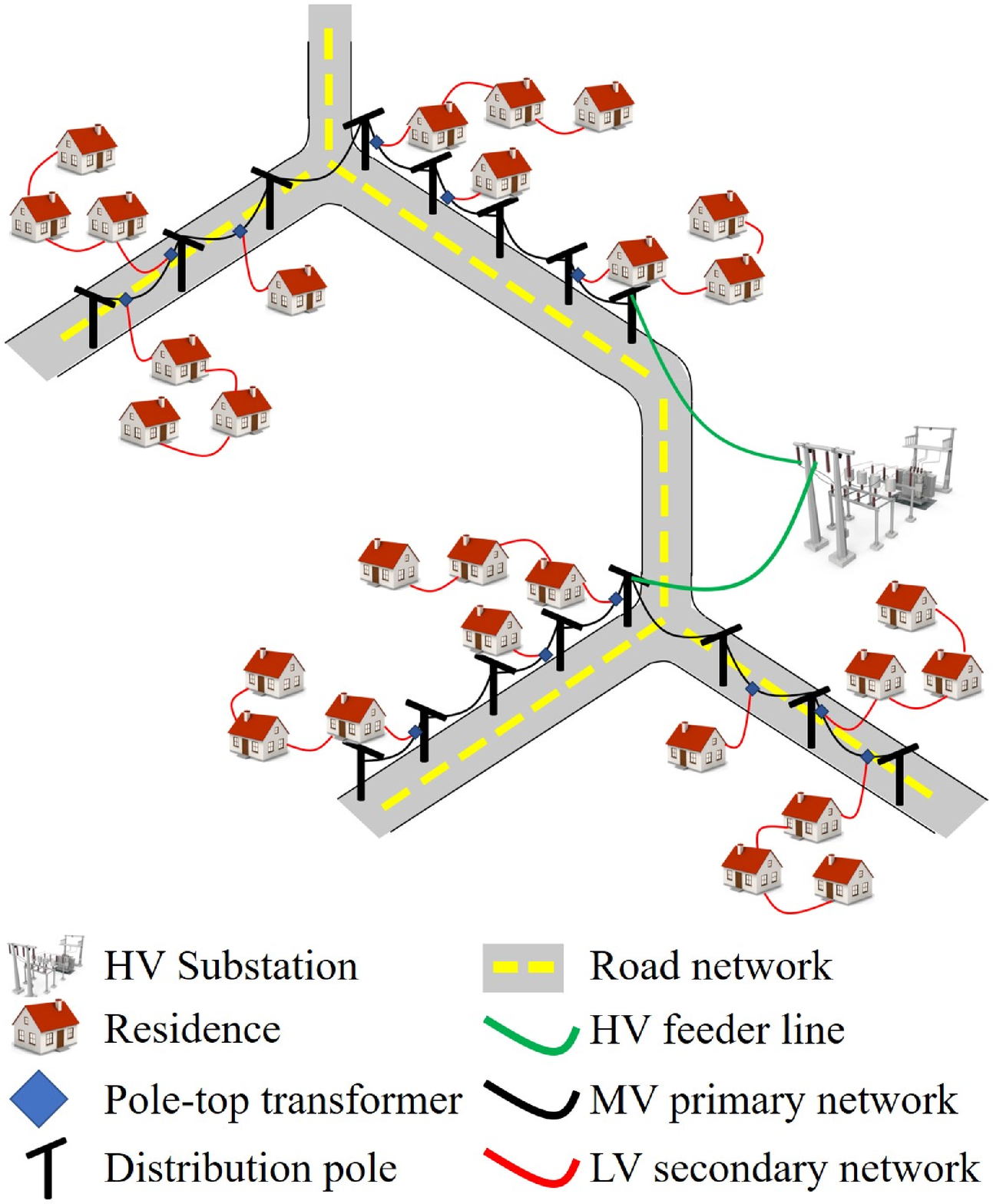}
	\includegraphics[width=0.235\textwidth]{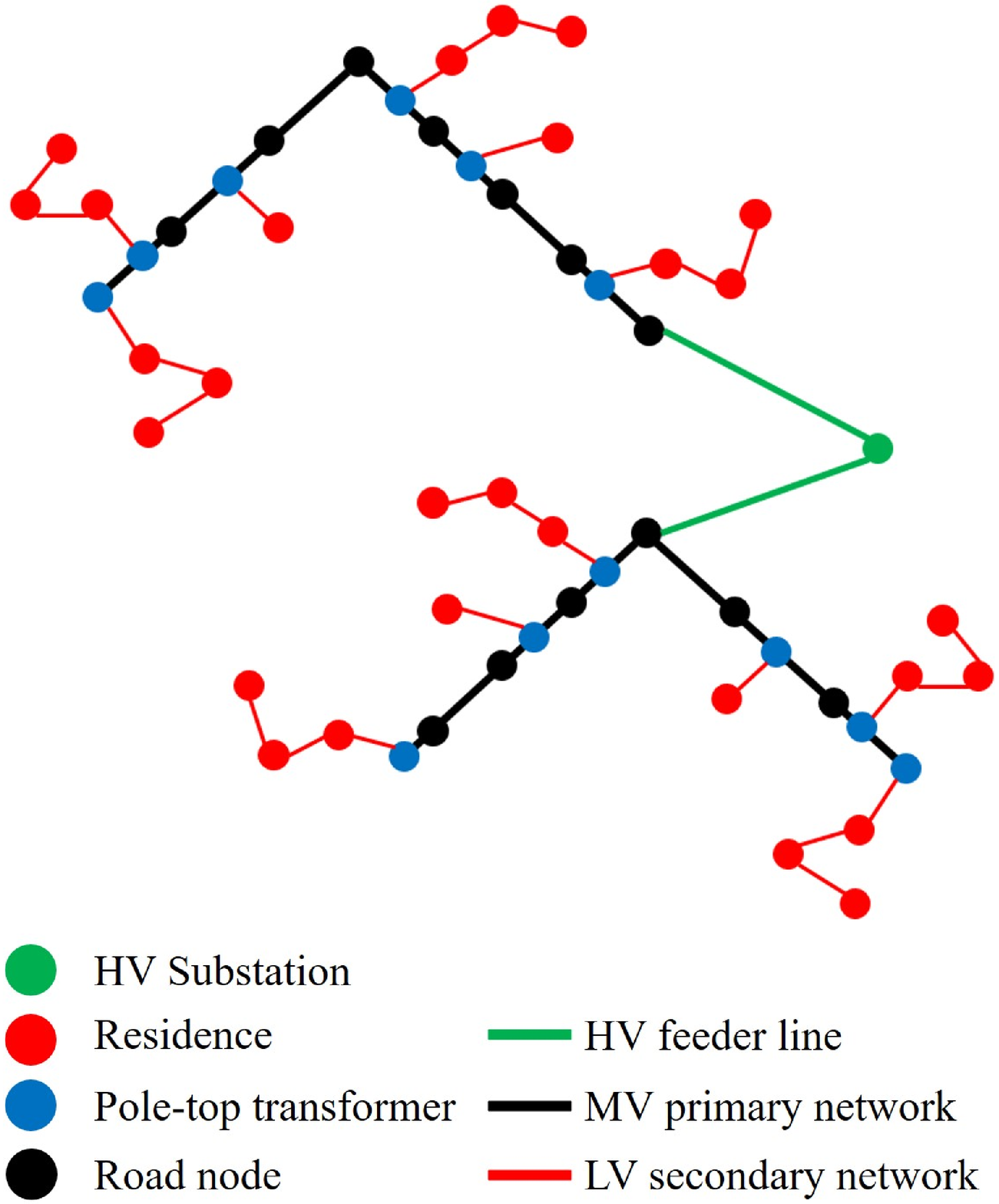}
	\caption{A schematic of synthetic distribution network. The substation feeds the distribution network through high voltage feeder lines. The primary network originates from the feeder lines and follows the road network as much as possible. The secondary network originates from transformers along primary network and connects individual residences. The single line diagram for the same network with its individual elements is shown on the right figure.}
	\label{fig:teaser}
\end{figure}
Distribution systems (primary and secondary) are usually configured in a \emph{radial} structure where power flow is unidirectional from source to consumers. Such radial structure facilitates protection coordination among reclosures, breakers and downstream fuses in the distribution system feeders~\cite{rad_prot}. Fig.~\ref{fig:teaser} shows a schematic of the structure of synthetic distribution network. The network connects substation to individual residential buildings through primary and secondary networks. The substation first feeds the primary network through high voltage feeder lines. It is assumed that the primary network follows the available road network to the maximum extent since distribution poles are generally placed along the road network to facilitate maintenance. The secondary network originates from the pole top transformers which are placed along primary network. The secondary network connects individual residences through short chains. Branching in these chains is avoided in order to reduce voltage drop at the leaf nodes.

\subsection{Available datasets}\label{ssec:data}
We use open source publicly available information regarding several infrastructures to generate the synthetic distribution networks. These data pertain to following sources: (i) transportation network data published by~\cite{navteq}, (ii) Geographical location of HV substations from data sets published by~\cite{eia_substations} and (iii) Residential electric power demand information developed in the models by~\cite{swapna_2018}.
\begin{table}[htb]
	\centering
	\caption{Datasets and related attributes used to generate synthetic distribution network}
	\label{tab:dataset-intro}
	\begin{small}
		\begin{tabular}{cllc}
			\toprule
			\textbf{Dataset} & \multicolumn{1}{c}{\textbf{Source}} & \multicolumn{1}{c}{\textbf{Attributes}} \\ \midrule
			Substation & \begin{tabular}[c]{@{}l@{}}Electric substation data\\ published by US \\Department of \\Homeland Security~\cite{eia_substations}\end{tabular} & \begin{tabular}[c]{@{}l@{}}\tabitem substation ID\\\tabitem longitude\\\tabitem latitude \end{tabular}\\ \midrule
			\begin{tabular}[c]{@{}c@{}}Road \\ network\end{tabular} & \begin{tabular}[c]{@{}l@{}}GIS and electronic \\navigable maps \\published by \\NAVTEQ~\cite{navteq}\end{tabular} & \begin{tabular}[c]{@{}l@{}}\tabitem node ID\\\tabitem node longitude\\\tabitem node latitude\\\tabitem link ID\\\tabitem link importance\\~~level \end{tabular} \\ \hline
			\multicolumn{1}{l}{Residences} & \begin{tabular}[c]{@{}l@{}}Synthetic population \\and electric load \\demand profiles\\ generated by~\cite{swapna_2018}\end{tabular} & \begin{tabular}[c]{@{}l@{}}\tabitem residence ID\\\tabitem longitude\\\tabitem latitude\\ \tabitem average \\~~hourly load \\~~demand\end{tabular} \\ \hline
		\end{tabular}
	\end{small}
\end{table}

\noindent\textbf{Roads} The road network represented in the form of a graph $\mathscr{R}=(\mathscr{V}_R,\mathscr{L}_R)$, where $\mathscr{V}_R$ and $\mathscr{L}_R$ are respectively the sets of nodes and links of the network. Each road link $l\in\mathscr{L}_R$ is represented as an unordered pair of terminal nodes $(u,v)$ with $u,v\in\mathscr{V}_R$. Each road node has a spatial embedding in form of longitude and latitude. 

\noindent\textbf{Substations} The set of $M$ substations $\mathscr{S}=\{s_1,s_2,\cdots,s_M\}$ with their respective geographical location data.

\noindent\textbf{Residences} The set of $N$ residential buildings with geographical coordinates $\mathscr{H}=\{h_1,h_2,\cdots,h_N\}$ and hourly load profile.

\noindent\textbf{Formal Problem Statement.}~Given a set of residence and electric substation locations, construct a power distribution network connecting these points such that (i) the network graph is a forest of trees originating from substations, (ii) covers all residential buildings, (iii) the edge flows and node voltages are within acceptable limits and (iv) overall length of the network is minimized. 

\section{Proposed Approach}\label{sec:approach}
The problem of creating a synthetic distribution network can be considered to be a combinatorial optimization problem, where we identify the optimal edges from a possible edge set connecting residences to substations. Due to the large number of points to be connected, the combinatorial problem is computationally challenging. As discussed in Section~\ref{sec:prelim}, a typical distribution system consists of primary and secondary networks. Therefore, the synthesis of such networks is considered to be a two-step bottom-up procedure: (i) the first step constructs the secondary distribution network connecting the residential buildings to pole-top/pad-mounted transformers and (ii) the second step involves connecting these local transformers to distribution substations through feeders and laterals. In our work, the road network is used as a proxy for the primary distribution network. Therefore, the locations of local pole-top transformers are considered to be internal points on the road network links.

The overall algorithm to generate synthetic distribution networks is outlined through the following points:
(i) Evaluate a mapping between residential buildings and road network links such that each residence is mapped to the nearest road link (\ref{ssec:map}). Thereafter, we identify probable locations of local distribution transformers along these road network links.
(ii) Connect the local distribution transformers to mapped residences in a radial configuration to resemble a typical secondary distribution network (\ref{ssec:secondary}).
(iii) Connect the local transformers to distribution substation in a forest of trees with the road network as a proxy (\ref{ssec:primary}).

\subsection{Mapping residences to road links}\label{ssec:map}
This section details the methodology to identify subsets of residential buildings near each road network link. This information would be used in the succeeding step to generate the secondary distribution network.
\begin{figure}[htbp]
	\centering
	\includegraphics[width=0.24\textwidth]{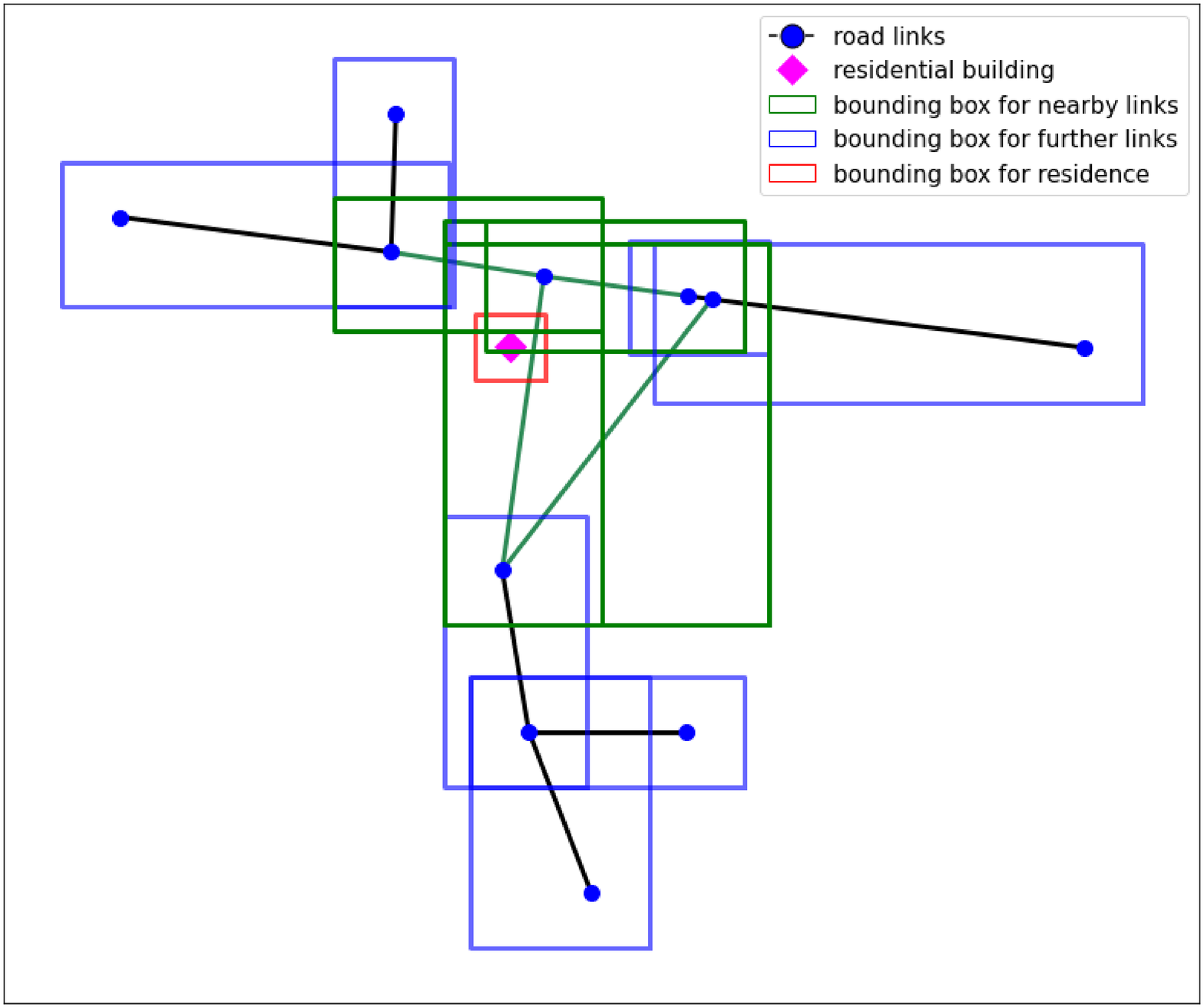}
	\includegraphics[width=0.24\textwidth]{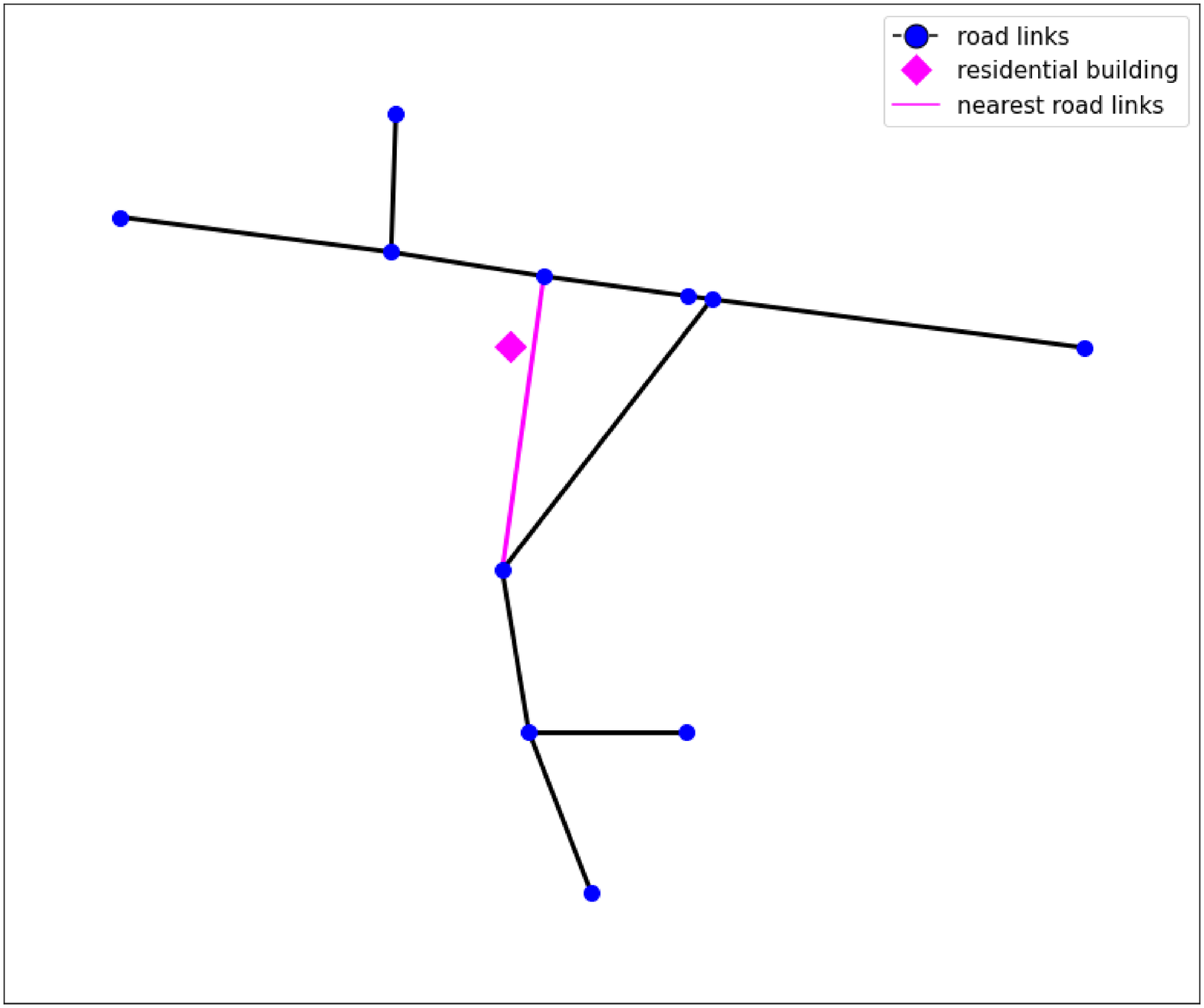}
	\caption{Steps showing procedure of Algorithm~\ref{alg:dist}. Bounding boxes are drawn around each road link. The nearby links to a given residence point are short-listed by identifying the boxes (green) which intersect with the one around the residence point (red). The nearest link is the closest among these shortlisted links.}
	\label{fig:mapping}
\end{figure}

Algorithm~\ref{alg:dist} lists the steps to compute the nearest road network link to a given point. Fig~\ref{fig:mapping} provides a visual representation of the same. First, we draw bounding boxes around each road link and residences. The intersections between the bounding box of a residence with those of the links are stored and indexed in a \emph{quad-tree} data structure. This information is retrieved to identify the links (green links in Fig.\ref{fig:mapping}) which are comparably nearer to the residential building than the others. Now, the geodesic distance of the residence point can be computed from these short-listed links and the nearest link can be identified henceforth. This approach reduces the computational burden of evaluating the distance between all road links and residential buildings.

\begin{algorithm}
	\caption{Find the nearest link in $\mathscr{L}_R$ to a given point $\mathbf{p}$.}
	\label{alg:dist}
	\begin{algorithmic}[1]
		\REQUIRE Dimensions of bounding box.
		\STATE Draw bounding box $\mathbf{B_l}$ around each link $\mathbf{l}\in\mathscr{L}_R$.
		\STATE Draw bounding box $\mathbf{B_p}$ around point $\mathbf{p}$.
		\STATE Find the bounding boxes $\mathbf{B_{l_1}},\mathbf{B_{l_2}},\cdots,\mathbf{B_{l_k}}$ corresponding to the links $\mathbf{l_1},\mathbf{l_2},\cdots,\mathbf{l_k}$ which intersect with $\mathbf{B_p}$.
		\STATE Find the link $\mathbf{l^\star}$ among $k$ short-listed links, which is nearest to point $\mathbf{p}$.
	\end{algorithmic}
\end{algorithm}

Following this step, we can compute an inverse map which represents the nearby residences mapped to a given road link. It is assumed that the residences would be fed by local pole top transformers located along the mapped link. However, their locations are not known beforehand. Therefore, we consider their probable locations along the link by interpolating multiple points at equal separation along it (based on standard engineering practice~\cite{pansini}).

\subsection{Creation of Secondary Network}\label{ssec:secondary}
The goal of this step is to create the secondary distribution network connecting local transformers along a given road network link $l$ to the set of $\mathscr{V}_H$ residences mapped to it. The set of probable local transformers interpolated along the link is $\mathscr{V}_T$. We need to connect the set of residences $\mathscr{V}_H$ to actual transformer nodes $\mathscr{V}_T^\star\subseteq\mathscr{V}_T$ in a forest of \emph{starlike} trees with each tree rooted at a transformer node. 

We start with a candidate set $\mathscr{E}_D$ of edges between the residence and probable transformer nodes. The secondary network edges are to be selected from this candidate set. We construct an undirected graph $\mathscr{G}_D:=(\mathscr{V},\mathscr{E}_D)$ with node set $\mathscr{V}=\mathscr{V}_T\bigcup\mathscr{V}_H$ and edge set $\mathscr{E}_D$. The aim of this step can be formalized through the following problem statement.

\begin{problem}[Secondary network creation problem]
	Given the undirected graph $\mathscr{G}_D(\mathscr{V},\mathscr{E}_D)$, find $\mathscr{E}_D^{\star}\subseteq\mathscr{E}_D$ such that the induced subgraph network $\mathscr{G}_D^\star(\mathscr{V}^\star,\mathscr{E}_D^\star)$ with $\mathscr{V}^\star=\mathscr{V}_T^\star\bigcup\mathscr{V}_H$ is a forest of starlike trees, $\mathscr{V}_T^\star\subseteq\mathscr{V}_T$ is the set of root nodes and the overall length of network is minimized.
\end{problem}
Note that a complete graph composed of the residence and transformer nodes can always be considered as the candidate set $\mathscr{E}_D$. In this paper, a Delaunay triangulation~\cite{delaunay} of the residential nodes is considered to reduce size of the problem.

\noindent\textbf{Edge weight assignment} An edge $e:(i,j)\in\mathscr{E}_D$ is assigned a weight $w(i,j)$
\begin{equation}
	w(i,j)=
	\begin{cases}
		\infty,\quad\quad\quad\quad\quad\quad\quad\quad\textrm{if }i,j\in\mathscr{V}_T\\
		\mathsf{dist}(i,j)+\lambda \mathsf{C}(i,j),~\textrm{otherwise}
	\end{cases}
	\label{eq:weight}
\end{equation}
where $\mathsf{dist}:\mathscr{V}\times\mathscr{V}\rightarrow\mathbb{R}$ denotes the geodesic distance between the nodes $i,j$. The function $\mathsf{C}(u,v)$ penalizes the cost function if the edge connecting nodes $u,v$ crosses the road link and is defined as
$$
\mathsf{C}(u,v)=
\begin{cases}
0,\quad \textrm{if }u,v\textrm{ are on same side of link}\\
2,\quad \textrm{if }u,v\textrm{ are on opposite side of link}\\
1,\quad \textrm{if }u\in\mathscr{V}_R\textrm{ or }v\in\mathscr{V}_R
\end{cases}
$$
$\lambda$ is a weight factor to penalize multiple crossing of edges over the road links. It also penalizes multiple edges emerging from the root node. The weights are stacked in $|\mathscr{E}_D|$-length vector $\mathbf{w}$. Note that an edge between two probable transformer nodes is assigned a weight of \emph{infinity} which is equivalent of not considering them as candidate edges.

\noindent\textbf{Edge variables} We introduce binary variables $\{x_e\}_{e\in\mathscr{E}_D}\in\{0,1\}^{|\mathscr{E}_D|}$. Variable $x_e=1$ indicates that the edge is present in the optimal topology and vice versa.
Each edge $e:=(i,j)$ is assigned a flow variable $f_e$ (arbitrarily) directed from node $i$ to node $j$. The binary variable and flows can be stacked in $|\mathscr{E}_D|$-length vectors $\mathbf{x}$ and $\mathbf{f}$ respectively.

\noindent\textbf{Node variables} The average hourly load demand at the $i^{th}$ residence node is denoted by $p_i$ and is strictly positive. We stack these average hourly load demands at all residence nodes in a $|\mathscr{V}_H|$ length vector $\mathbf{p}$.

\noindent\textbf{Degree constraint.}~Statistical surveys on distribution networks in~\cite{mv_2011,review2017} show that residences along the secondary network are mostly connected in series with at most $2$ neighbors. This is ensured by (\ref{eq:sec-degree}) which limits degree of residence nodes to $2$.
\begin{equation}
	\sum_{e:(h,j)}x_{e}\leq 2,\quad \forall h\in\mathscr{V}_{H}\label{eq:sec-degree}
\end{equation}

\noindent\textbf{Power flow constraints.}~For the connected graph $\mathscr{G}_D(\mathscr{V},\mathscr{E}_D)$, we define the $|\mathscr{E}_D|\times|\mathscr{V}|$ branch-bus incidence matrix $\mathbf{A}_{\mathscr{G}_D}$ with the entries as
\begin{equation}
	\begin{matrix}
		\mathsf{A}_{\mathscr{G}_D}(e,k):=
		\begin{cases}
			~~1,\quad k=i\\-1,\quad k=j\\~~0,\quad\textrm{otherwise}
		\end{cases}
		&\forall{e=(i,j)\in\mathscr{E}_D}
	\end{matrix}
	\label{eq:bus-incidence}
\end{equation}
Since the order of rows and columns in $\mathbf{A}_{\mathscr{G}_D}$ is arbitrary, we can partition the columns as $\mathbf{A}_{\mathscr{G}_D}=\begin{bmatrix}\mathbf{A_{T}}&\mathbf{A_{H}}\end{bmatrix}$, without loss of generality, where the partitions are the columns corresponding to transformer and residence nodes respectively. We define $\mathbf{A_{H}}$ as the reduced branch-bus incidence matrix.

Assuming no network losses, (\ref{eq:pf-sec1}) represents the power balance equations at all residence nodes. Note that the optimal network is obtained from $\mathscr{G}_D$ after removing the edges for which $x_e=0$. Therefore, we need to enforce zero flows $f_e$ for non-existing edges. The constraint (\ref{eq:pf-sec2}) performs this task along with constraining the flows $f_e$ for existing edges to be within pre-specified capacities $\overline{f}$.
\begin{subequations}
	\begin{align}
		&\mathbf{A_H}^T\mathbf{f}=\mathbf{p}\label{eq:pf-sec1}\\
		-&\overline{f}\mathbf{x}\leq \mathbf{f}\leq \overline{f}\mathbf{x}\label{eq:pf-sec2}
	\end{align}
	\label{eq:pf-sec}
\end{subequations}

\noindent\textbf{Ensuring radial topology} The radiality requirement of $\mathscr{G}_D^\star$ can be enforced from a known graph theory property: \emph{a forest with $n$ nodes and $m$ root nodes has $n-m$ edges}. In our case, $|\mathscr{V}_H|+|\mathscr{V}_T^\star|$ nodes need to be covered in a forest of trees with $|\mathscr{V}_T^\star|$ root nodes which leads us to the following constraint.
\begin{equation}
	\sum_{e\in\mathscr{E}_D}x_e=|\mathscr{V}_H|\label{eq:radial}
\end{equation}

However, we need to ensure that there are no disconnected cycles in the optimal network. This can be done by ensuring that the residence points are connected to a transformer node~\cite{manish2019,lei2019radiality}. In our case, this condition is satisfied by the node power flow condition in (\ref{eq:pf-sec1}) if all the residential nodes consume power. This is an extension of the proposition in~\cite{manish2019} which considers renewable generation at the nodes. However, contrary to the previous work, the following proposition is a special case where every non-root node has only positive load demand. Using the power balance constraints to ensure radial topology reduces the number of constraints when dealing with large sized networks.

The optimal secondary network graph $\mathscr{G}_D^\star(\mathscr{V}^\star,\mathscr{E}_D^\star)$ with $\mathscr{E}_D^\star=\mathscr{E}_D\setminus\{e:e\in\mathscr{E}_D, x_e=0\}$ and $\mathscr{V}^\star=\mathscr{V}_T^\star\bigcup\mathscr{V}_H$ has at least $|\mathscr{V}_T^\star|$ components since local transformer nodes can only be root nodes and there is no candidate edge between them.
\begin{proposition}
	The graph $\mathscr{G}_D^\star(\mathscr{V}^\star,\mathscr{E}_D^\star)$ with reduced branch-bus incidence matrix $\mathbf{A_H}$ (corresponding to columns of $\mathscr{V}_H$) and node power demand vector $\mathbf{p}\in\mathbb{R}^{\mathscr{V}_H}$, with strictly positive entries, has exactly $|\mathscr{V}_T^\star|$ connected components if and only if there exists $\mathbf{f}\in\mathbb{R}^{|\mathscr{E}_D^\star|}$, such that (\ref{eq:pf-sec1}) is satisfied.
	\label{prop-1}
\end{proposition}
\begin{proof}
	The proof follows along similar lines to proof of Proposition~1 in~\cite{manish2019}. Proving by contradiction, suppose $\mathscr{G}_D^\star(\mathscr{V}^\star,\mathscr{E}_D^\star)$ has more than $|\mathscr{V}_T^\star|$ connected components and there exists $\mathbf{f}\in\mathbb{R}^{|\mathscr{E}_D^\star|}$ satisfying the proposed equality. Therefore, there exists a connected component $\mathscr{G}_{\mathscr{S}}(\mathscr{V}_{\mathscr{S}},\mathscr{E}_{\mathscr{S}})$ which is a maximal connected subgraph with $\mathscr{V}_{\mathscr{S}}\subset\mathscr{V}_H$ and $\mathscr{V}_{\mathscr{S}}\bigcap\mathscr{V}_T^\star=\emptyset$. Let $\mathbf{A}_\mathscr{S}$ denote the bus incidence matrix of ${\mathscr{G}}_{\mathscr{S}}$. By definition, it holds that $\mathbf{A}_\mathscr{S}\mathbf{1}=\mathbf{0}$. 
	
	Since graph $\mathscr{G}_{\mathscr{S}}(\mathscr{V}_{\mathscr{S}},\mathscr{E}_{\mathscr{S}})$ is a maximal connected subgraph of $\mathscr{G}_D^\star$, there exists no edge $(i,j)$ with $i\in\mathscr{V}_\mathscr{S}$ and $j\in\mathscr{V}_{\overline{\mathscr{S}}}$, where $\mathscr{V}_{\overline{\mathscr{S}}}=\mathscr{V}^\star\setminus\mathscr{V}_{\mathscr{S}}$. Since the order of rows and columns of $\mathbf{A_H}$ are arbitrary, we can partition without loss of generality as
	\begin{equation}
		\mathbf{A_H}=\begin{bmatrix}\mathbf{A}_{\overline{\mathscr{S}}}&\mathbf{0}\\\mathbf{0}&\mathbf{A}_{\mathscr{S}}\end{bmatrix}\notag
	\end{equation}
	We can partition vectors $\mathbf{f}$ and $\mathbf{p}$ conformably to $\mathbf{A_H}$ to get the following equality 
	\begin{equation}
		\begin{bmatrix}\mathbf{A}_{\overline{\mathscr{S}}}&\mathbf{0}\\\mathbf{0}&\mathbf{A}_{\mathscr{S}}\end{bmatrix}^T\begin{bmatrix}\mathbf{f}_{\overline{\mathscr{S}}}\\\mathbf{f}_{\mathscr{S}}\end{bmatrix}=\begin{bmatrix}\mathbf{p}_{{\overline{\mathscr{S}}}}\\\mathbf{p}_{\mathscr{S}}\end{bmatrix}\notag
	\end{equation}
	where $\mathbf{p}_{\mathscr{S}},\mathbf{p}_{{\overline{\mathscr{S}}}}$ are respectively the vectorized power demands for node sets $\mathscr{V}_\mathscr{S}$ and $i\in\mathscr{V}_{\overline{\mathscr{S}}}$. From the second block, it implies that 
	\begin{equation}
		\mathbf{A}_{\mathscr{S}}^T\mathbf{f}_{\mathscr{S}}=\mathbf{p}_{\mathscr{S}}\Rightarrow \mathbf{1}^T\mathbf{p}_{\mathscr{S}}=\mathbf{1}^T\mathbf{A}_{\mathscr{S}}^T\mathbf{f}_{\mathscr{S}}=\mathbf{0}\label{eq:claim}
	\end{equation}
	Since all entries of $\mathbf{p}$ are strictly positive from the initial assumption, we have $\mathbf{1}^T\mathbf{p}_{\mathscr{S}}\neq0$ which contradicts (\ref{eq:claim}) and completes the proof. 
\end{proof}

\noindent\textbf{Generating optimal network topology}
The optimal secondary network is obtained after solving the optimization problem.
\begin{equation}
	\begin{aligned}
		\min_{\mathbf{x}}&~\mathbf{w}^T\mathbf{x}\\
		\textrm{s.to.}&~(\ref{eq:sec-degree}),(\ref{eq:pf-sec}),(\ref{eq:radial})
	\end{aligned}
\end{equation}

\begin{figure}[htbp]
	\centering
	\includegraphics[width=0.48\textwidth]{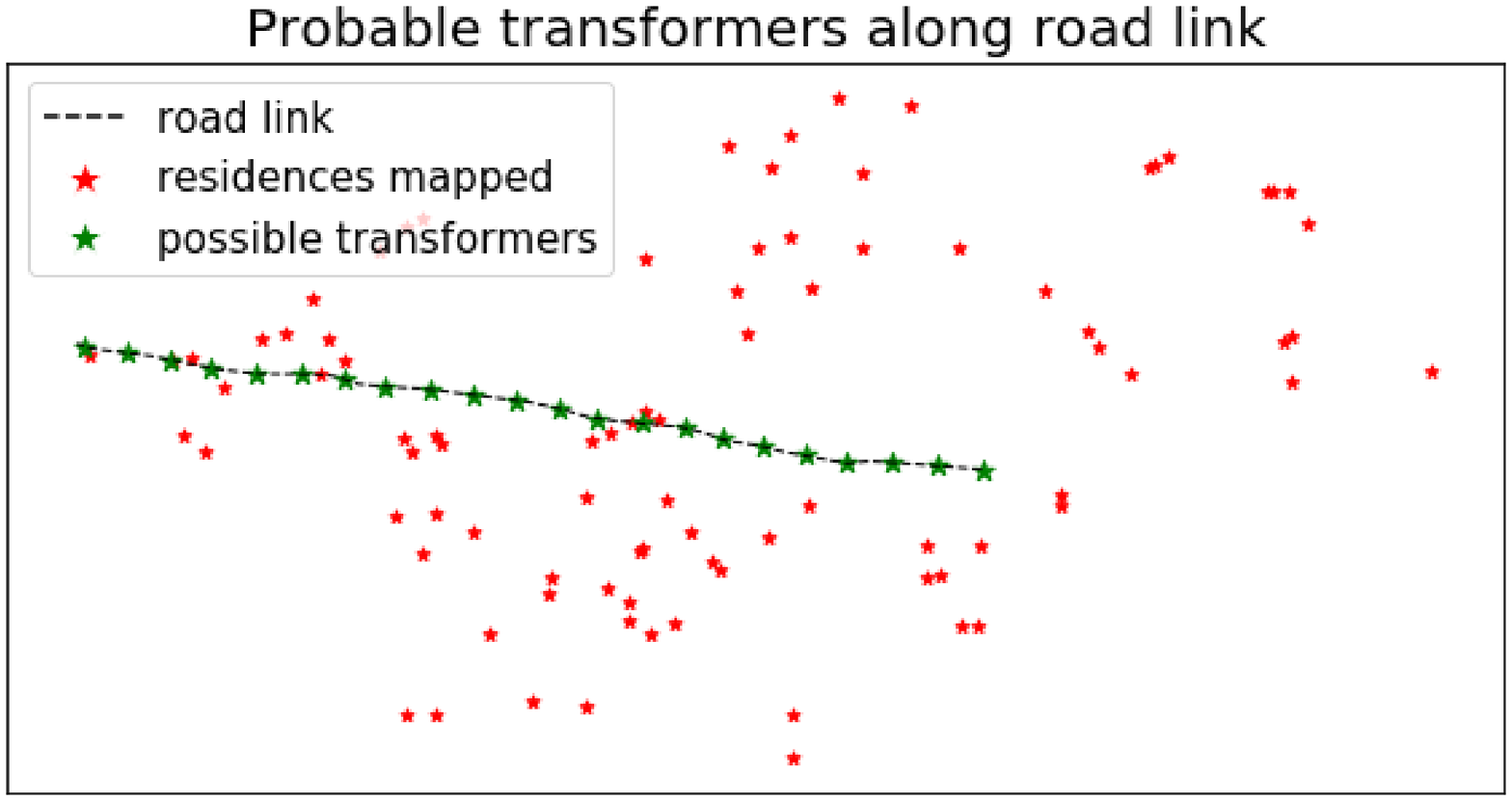}
	\includegraphics[width=0.48\textwidth]{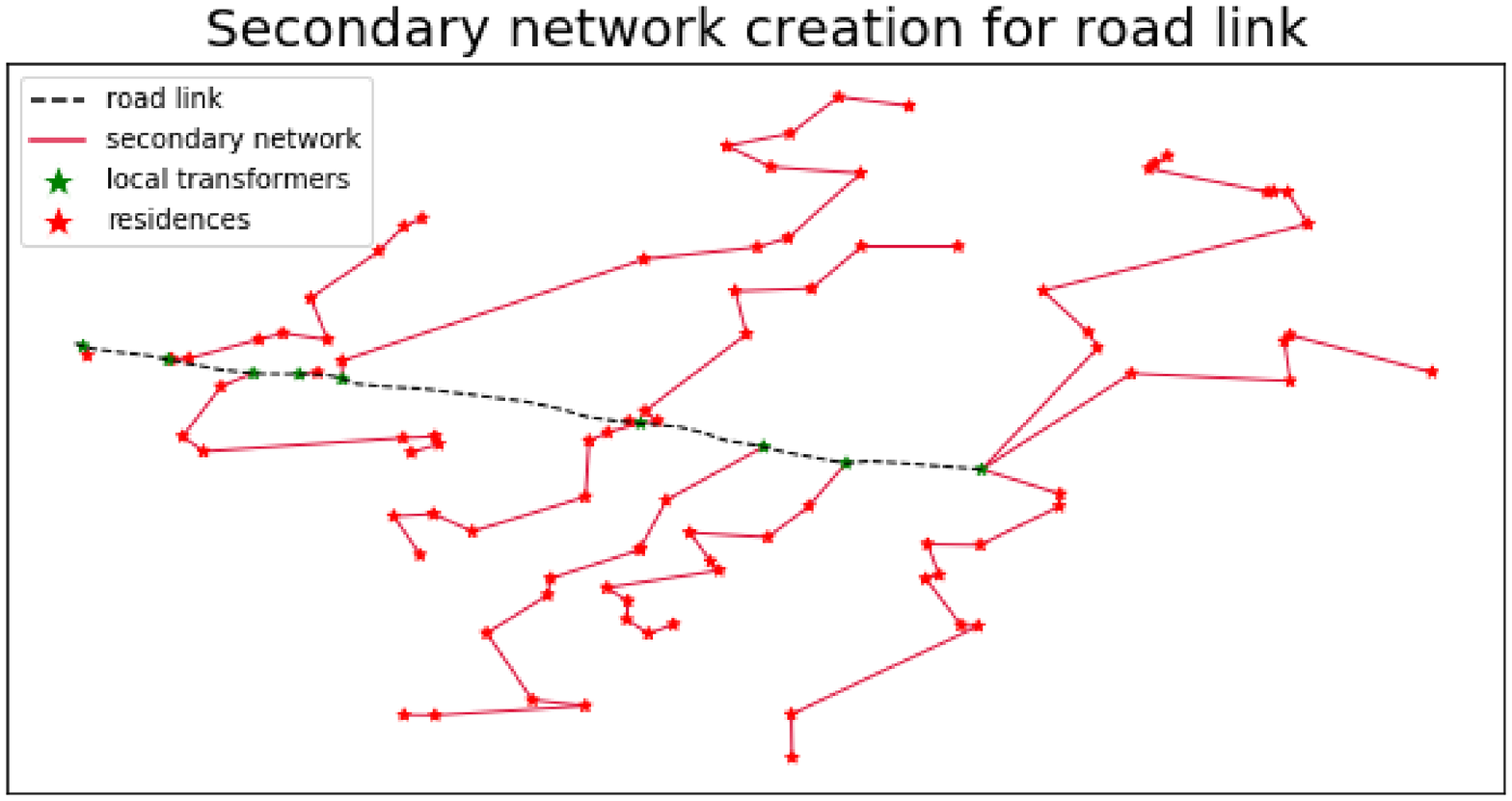}
	\caption{The probable transformer locations are identified along the road network link and the secondary network is created by solving the optimization problem. The network originates from the road link and connects residences mapped to it in a forest of starlike trees rooted at the transformers.}
	\label{fig:secondary}
\end{figure}
\noindent\textbf{Example.}~The secondary network connecting residences along a road network link is shown in Fig.~\ref{fig:secondary}. The first figure shows the residences (red) mapped to the road link and the probable transformers (green points) along the link. A Delaunay triangulation is considered to connect the points and obtain the set of possible edges. Thereafter, the optimization problem is solved which identifies the optimal set of edges as shown in the bottom figure. We observe that all residences are connected in star-like trees with roots as one of the transformers along the road link.

\subsection{Creation of Primary Network}\label{ssec:primary}
In this step, the goal is to create the primary distribution network which connects the substations to the local transformers. We aggregate the load at residences to the local transformer locations (obtained as an output from the preceding step). In this paper, we have aggregated the average demand at the residences; however, \emph{diversity factor} may be used in the aggregation to appropriately rate the transformers~\cite{kersting_book}. The objective of the current step is to connect all these local transformer locations to the substation.

\begin{figure}[htbp]
	\centering
	\includegraphics[width=0.48\textwidth]{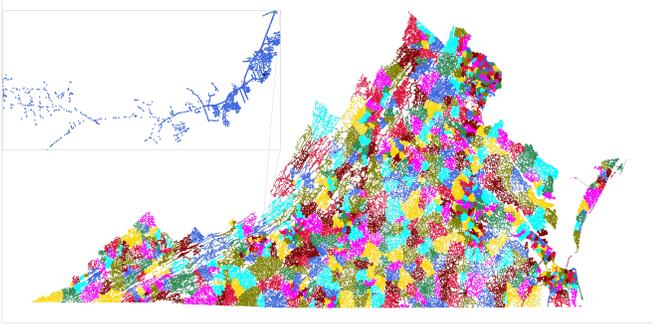}
	\caption{Voronoi partitioning of the road and transformer nodes in Virginia based on the shortest path distance (along road network) to the substations.}
	\label{fig:partition-1}
\end{figure}
\noindent\textbf{Partitioning the problem.}~As mentioned earlier in Section~\ref{sec:prelim}, we use the road network as a proxy for the synthetic primary network. The problem of primary network generation is to select the optimal edge set from the candidate road links. However, the combinatorial problem becomes challenging to solve when considering a large geographical area. Therefore, the problem of generating the primary network is partitioned into a number of sub-problems (one problem corresponding to each distribution substation). First, we include the local transformer nodes (located along road links) in the original road network graph. Thereafter, we partition the nodes in an intermediate step and use the road network subgraph composed of these partitioned nodes to construct the primary distribution network. 

The nodes can be partitioned such that each of them is assigned to the geographically nearest substation. However, such partitioning gives rise to multiple disconnected components (of varied sizes) in the induced subgraph. In such a situation, the transformers cannot be connected through the road network. To this end, the transformers are clustered so that each of them is mapped to the nearest substation along the road network. This is depicted in Fig~\ref{fig:partition-1} where each color represents a partition of transformer nodes. These partitions are known as Voronoi cells which are centered at the substation location. The partitioning is done based on the \emph{shortest path distance} metric~\cite{voronoi} which ensures that each node is mapped to the nearest substation along the road network and each induced subgraph has a single connected component.

\noindent\textbf{Requirements} Once we obtain the \emph{road network subgraph} $\mathscr{G}_R(\mathscr{V},\mathscr{E}_R)$ with $\mathscr{V}=\mathscr{V}_R\bigcup\mathscr{V}_T$ for each partition, the goal of the current step is to select edges $\mathscr{E}_R^{\star}\subseteq\mathscr{E}_R$ such that (i) all local transformer nodes ($i\in\mathscr{V}_T$) are connected, (ii) the road network nodes ($i\in\mathscr{V}_R$) can be covered only to connect local transformer nodes (i.e., the road nodes cannot be leaf nodes in the created primary network) and (iii) the created primary network is a forest of trees with each tree being rooted at a road node and connects local transformer nodes through the road links in a tree structure.

The last condition follows the assumption that a substation can have multiple feeder lines connecting several root nodes feeding local transformers in remote localities. The formal problem statement is provided here.
\begin{problem}[Primary network creation problem]
	Given the connected network $\mathscr{G}_R(\mathscr{V},\mathscr{E}_R)$, find $\mathscr{E}_R^{\star}\subseteq\mathscr{E}_R$ such that the induced subgraph network $\mathscr{G}_R^\star(\mathscr{V}^\star,\mathscr{E}_R^\star)$, with $\mathscr{V}^\star=\mathscr{V}_R^\star\bigcup\mathscr{V}_T$, is a forest of trees with each tree rooted at some $r\in\mathscr{V}_R^\star\subseteq\mathscr{V}_R$ and the overall network length is minimized.
\end{problem}

\noindent\textbf{Node variables.}~The node set $\mathscr{V}=\mathscr{V}_R\bigcup\mathscr{V}_T$ comprises of road and transformer nodes respectively. The residential average hourly demands are aggregated at the transformer locations and stacked in a $|\mathscr{V}_T|$-length vector $\mathbf{p}$. Note that the power demand at road nodes is zero since they are dummy nodes with a purpose to connect transformer nodes. Let $v_i$ represent the voltage at the node $i$. The node voltages can be stacked in $|\mathscr{V}_R|+|\mathscr{V}_T|$ length vectors $\mathbf{v}$.

We assign binary variables $\{y_r,z_r\}_{r\in\mathscr{V}_R}\in\{0,1\}$. Variable $y_r=1$ indicates that road network node $r$ is part of the primary network and vice versa. Variable $z_r=0$ indicates that road network node $r$ is a root node (connected to substation through a high voltage feeder line) in the primary network and $z_r=1$ indicates otherwise. The binary variables are stacked in $|\mathscr{V}_R|$-length vectors $\mathbf{y}$ and $\mathbf{z}$ respectively.

\noindent\textbf{Edge variables}
In order to identify which edges comprise of the primary network, each edge $e$ is assigned a binary variable $\{x_e\}_{e\in\mathscr{E}_R}\in\{0,1\}^{|\mathscr{E}_R|}$. Variable $x_e=1$ indicates that the edge is part of the optimal primary network and vice versa. The binary variable for all edges and edge power flows can be stacked in a $|\mathscr{E}_R|$-length vectors $\mathbf{x}$ and $\mathbf{f}$ respectively.

\noindent\textbf{Connectivity constraint}
Let $e=(r,j)\in\mathscr{E}_R$ denote an edge which is incident on the road node $r\in\mathscr{V}_{R}$. $\sum_{e:(r,j)}x_{e}$ is the degree of node $r$ in graph $\mathscr{G}_R$.
\begin{subequations}
	\begin{align}
		\sum_{e=(r,j)}&x_{e}\leq |\mathscr{E}_R|y_r, &\forall r\in\mathscr{V}_{R}\label{eq:non-chosen}\\
		\sum_{e=(r,j)}&x_{e}\geq y_r, &\forall r\in\mathscr{V}_{R}\label{eq:chosen}\\
		1-&z_r\leq y_r, &\forall r\in\mathscr{V}_{R}\label{eq:non-root}\\
		\sum_{e:(r,j)}&x_{e}\geq2(y_r+z_r-1), &\forall r\in\mathscr{V}_{R}\label{eq:transfer}
	\end{align}
	\label{eq:prim-connectivity}
\end{subequations}

\noindent If $r$ is not included ($y_r=0$) in the primary network, (\ref{eq:non-chosen}) and (\ref{eq:chosen}) ensures that there are no incident edges on $r$ since $\sum_{e=(r,j)}x_{e}=0$. Further, (\ref{eq:non-root}) ensures $z_r=1$ which implies that an unselected road node is treated as a non-root node.

\noindent If $r$ is included in primary network and is not a root node ($y_r=1,z_r=1$), it has to be a transfer node with a minimum degree of $2$ which is ensured by (\ref{eq:transfer}) through $\sum_{e=(r,j)}x_{e}\geq2$.

\noindent If $r$ is included in primary network and is a root node ($y_r=1,z_r=0$), (\ref{eq:non-chosen}), (\ref{eq:chosen}) and (\ref{eq:transfer}) ensures that the degree of the road node is positive. 

\noindent\textbf{Ensuring radiality constraint.}~ The primary network is a forest of trees with $|\mathscr{V}_T|+\sum_{r\in\mathscr{V}_R}y_r$ nodes and $\sum_{r\in\mathscr{V}_R}(1-z_r)$ root nodes. Therefore, from graph theory, we have the radiality constraint as 
\begin{equation}
	\sum_{e\in\mathscr{E}_R}x_e=|\mathscr{V}_T|+\sum_{r\in\mathscr{V}_{R}}y_r-\sum_{r\in\mathscr{V}_{R}}(1-z_r)\label{eq:prim_radial}
\end{equation}
However, this is not a sufficient condition for radiality since it does not avoid formation of disconnected components with cycles. Therefore, we need to enforce constraints to avoid formation of cycles in individual trees.

\noindent\textbf{Power balance and flow constraints}
We can define the branch-bus incidence matrix $\mathbf{A_{\mathscr{G}_R}}\in\mathbb{R}^{|\mathscr{E}_R|\times|\mathscr{V}|}$. Since the order of rows and columns in $\mathbf{A_{\mathscr{G}_R}}$ is arbitrary, we can partition the rows without loss of generality as $\mathbf{A_{\mathscr{G}_R}}=\begin{bmatrix}\mathbf{A_{T}}&\mathbf{A_{R}}\end{bmatrix}$. Here, the partitions $\mathbf{A_{T}}$ and $\mathbf{A_{R}}$ are obtained by stacking the columns of $\mathbf{A_{\mathscr{G}_R}}$ corresponding to transformer and road nodes respectively.
\begin{subequations}
	\begin{align}
		&\mathbf{A_{T}}^T\mathbf{f}=\mathbf{p}\label{eq:pf-tsfr}\\
		-&\overline{s}(1-\mathbf{z})\leq\mathbf{A_{R}}^T\mathbf{f}\leq\overline{s}(1-\mathbf{z})\label{eq:pf-road}\\
		-&\overline{f}\mathbf{x}\leq \mathbf{f}\leq \overline{f}\mathbf{x}\label{eq:pf-flow}
	\end{align}
	\label{eq:prim-pf}
\end{subequations}
(\ref{eq:pf-tsfr}) ensures that a path exists between each transformer node and a root node. If a road node is not a root node (with $z_r=1$), (\ref{eq:pf-road}) enforces that the consumption at the node is $0$. If a road node is a root node (with $z_r=0$), the power consumption/injection at the node is limited by the feeder capacity $\overline{s}$. (\ref{eq:pf-flow}) ensures that if an edge is selected ($x_e=1$), the power flow is limited by the line capacity $\overline{f}$.

\noindent\textbf{Voltage constraints.}~ The created primary network should also satisfy the operational voltage constraints. The long HV lines from the substation to the root nodes in the optimal primary distribution network end in voltage regulators which ensure that the root nodes have a voltage of $1$pu. This is ensured by (\ref{eq:volt-root}). (\ref{eq:volt-limit}) limits the voltage at all nodes within the acceptable limits of $[\underline{v},\overline{v}]$.
\begin{subequations}
	\begin{align}
		(1-v_r)\leq &z_r~\forall r\in\mathscr{V}_{s_R}\label{eq:volt-root}\\
		\underline{v}\mathbf{1}\leq \mathbf{v}\leq &\overline{v}\mathbf{1}\label{eq:volt-limit}\\
		-(1-x_e)M\leq &v_i-v_j-r_ef_e\leq(1-x_e)M,~\forall e\in\mathscr{E}_s\label{eq:volt-edge}
	\end{align}
	\label{eq:prim-voltage}
\end{subequations}
If $r_e$ denotes the resistance of the line $e:(i,j)\in\mathscr{E}_R$, the Linearized Distribution Flow (LDF) model relates the squared voltage magnitude to power flows linearly as $v_i^2-v_j^2=2r_ef_e$ where $f_e$ is the entry from the vector $\mathbf{f}$ corresponding to edge $e$. The squared voltage can be approximated as $v_i^2\approx 2v_i-1$ which leads to the relation $v_i-v_j=r_eP_e$ (see~\cite{ldf}). Notice that (\ref{eq:volt-edge}) enforces this constraint on only those edges $e\in\mathscr{E}_R$ for which $x_e=1$. Here, a tight bound of $M$ can be selected as $M=\overline{v}-\underline{v}$.

\noindent\textbf{Generating optimal primary network} Each edge $e=(u,v)\in\mathscr{E}_R$ is assigned a weight $w_e=w(u,v)=\mathsf{dist}(u,v)$ which is the geodesic distance between the nodes. Additionally, for every road node $r\in\mathscr{V}_{R}$, we compute its geodesic distance from the substation $s$ and is denoted by $\mathsf{d}_r$. The optimal primary network topology is obtained by solving the optimization problem.
\begin{equation}
	\begin{aligned}
		\min_{\mathbf{x},\mathbf{y},\mathbf{z}}&~\sum_{e\in\mathscr{E}_s}x_ew_e+\sum_{e\in\mathscr{V}_{s_R}}(1-z_r)d_r\\
		\textrm{s.to.}&~(\ref{eq:prim-connectivity}),(\ref{eq:prim_radial}),(\ref{eq:prim-pf}),(\ref{eq:prim-voltage})
	\end{aligned}
\end{equation}

\noindent\textbf{Handling large optimization problems.}~Note that for the optimization problem in (\ref{eq:prim-pf}), each road node corresponds to two binary variables and each edge corresponds to one. For a Voronoi partition with a large number of road and transformer nodes, the optimization problem to create the primary network is required to be scaled for faster computation. For this purpose, we further partition the nodes based on \emph{Girvan-Newman} hierarchical community detection algorithm~\cite{girvan7821}. We continue partitioning until all the partitions reach a specified condition (for example, have at most some specified number of nodes, have at most a specified load demand etc.). Such a partitioning leads to connected components in the graph and the optimization problem in (\ref{eq:prim-pf}) can be solved for each partition separately.

\begin{figure*}[htbp]
	\centering
	\includegraphics[width=0.32\textwidth]{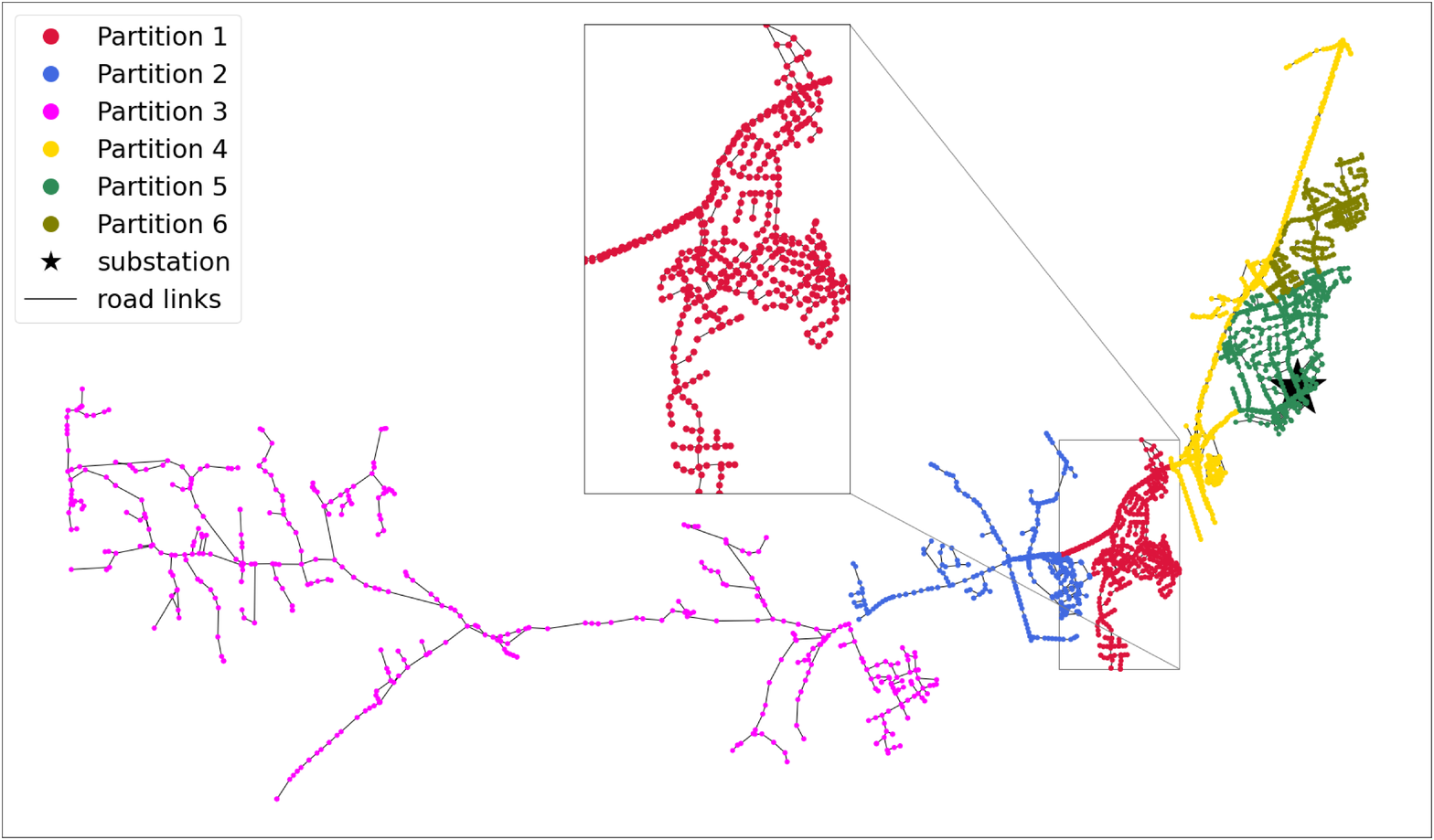}
	\includegraphics[width=0.32\textwidth]{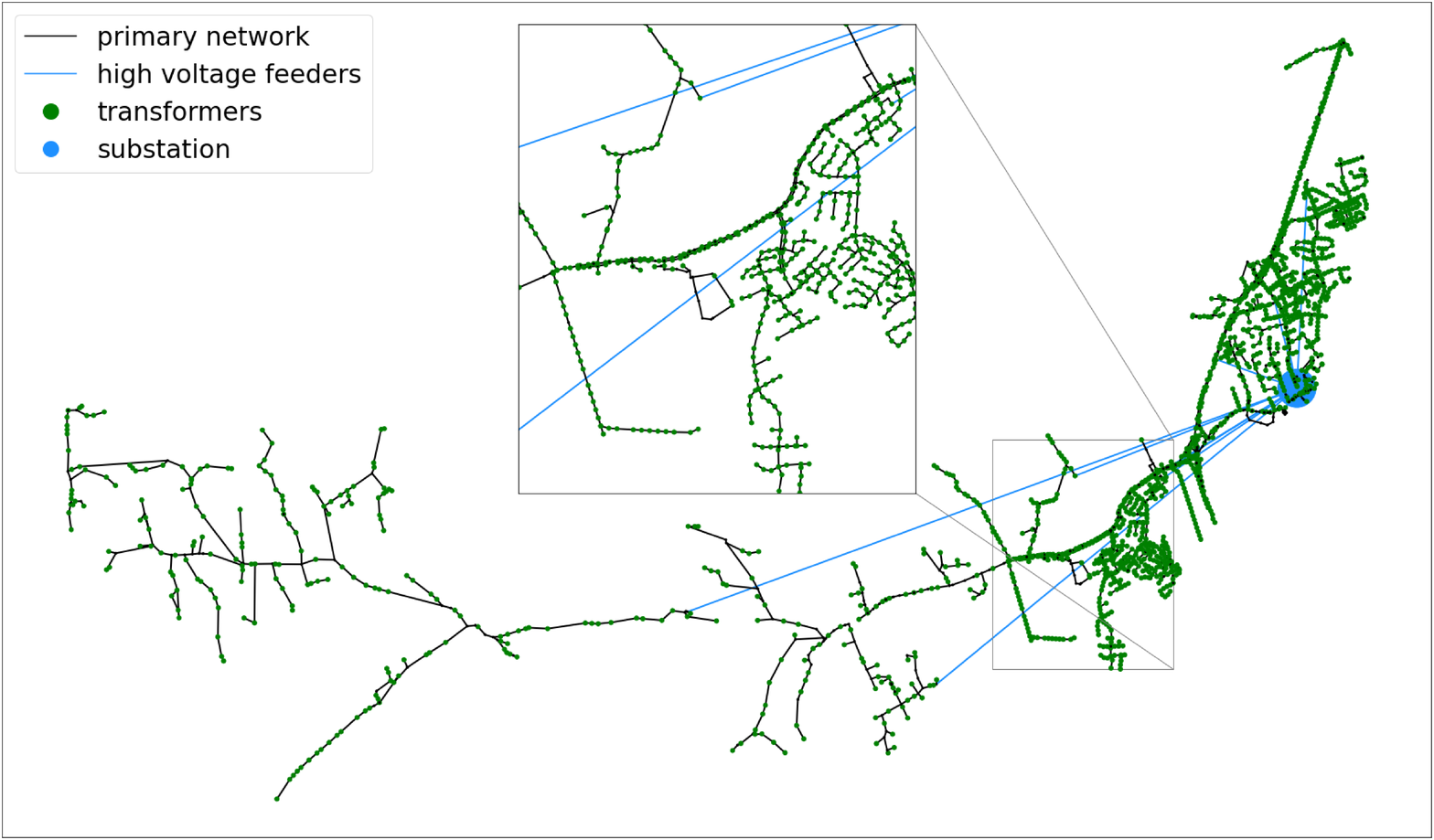}
	\includegraphics[width=0.32\textwidth]{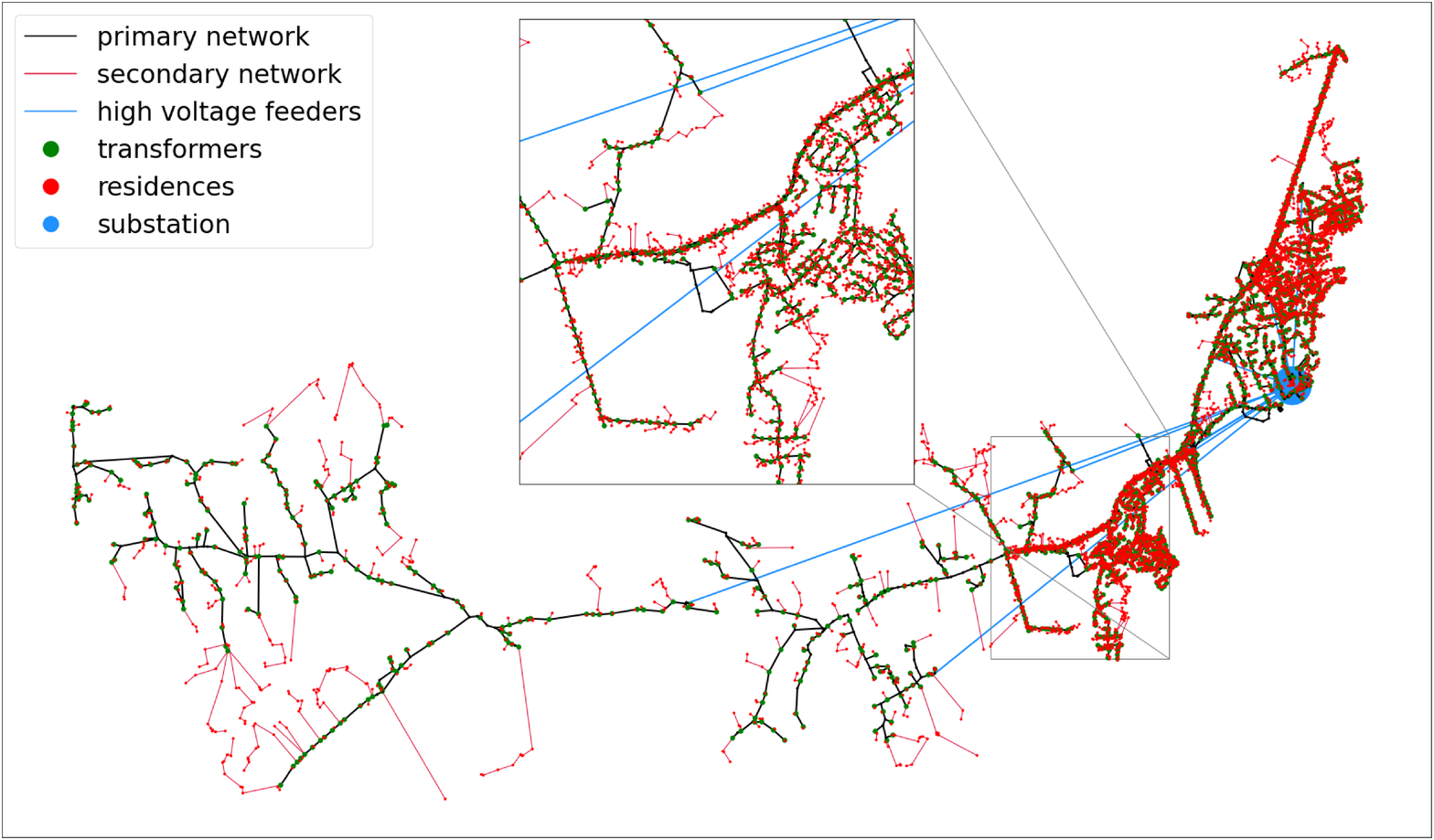}
	\caption{Synthetic primary network generation for a substation in south-west Virginia. First, the Girvan-Newman algorithm is applied to partition the large road network graph. Thereafter, the optimization problem is solved to obtain the optimal primary network. Finally the complete synthetic network is created by adding the secondary network to the primary.}
	\label{fig:prim-example}
\end{figure*}
\noindent\textbf{Example.}~We synthesize the primary distribution network originating from a substation in south-west Virginia. The induced subgraph of the road network composed of nodes in the partition consists of 2990 nodes and 3178 edges. We apply the Girvan-Newman hierarchical community detection algorithm to further partition the road network subgraph. Fig.~\ref{fig:prim-example} shows the partitions in the network which are of comparable sizes (less than 700 nodes in each partition). Thereafter, we solve the optimization problem (\ref{eq:prim-pf}) for each of these sub-partitions. Note that the substation connects these partitions through one or more high voltage feeder lines (blue lines). The optimal edges are selected from the road links such that all the transformer nodes (green nodes) are covered in a tree structure and the overall length of the network is minimized. Finally, the associated secondary distribution network is appended with the primary to obtain the complete synthetic distribution network.

\begin{figure}[htbp]
	\centering
	\includegraphics[width=0.24\textwidth]{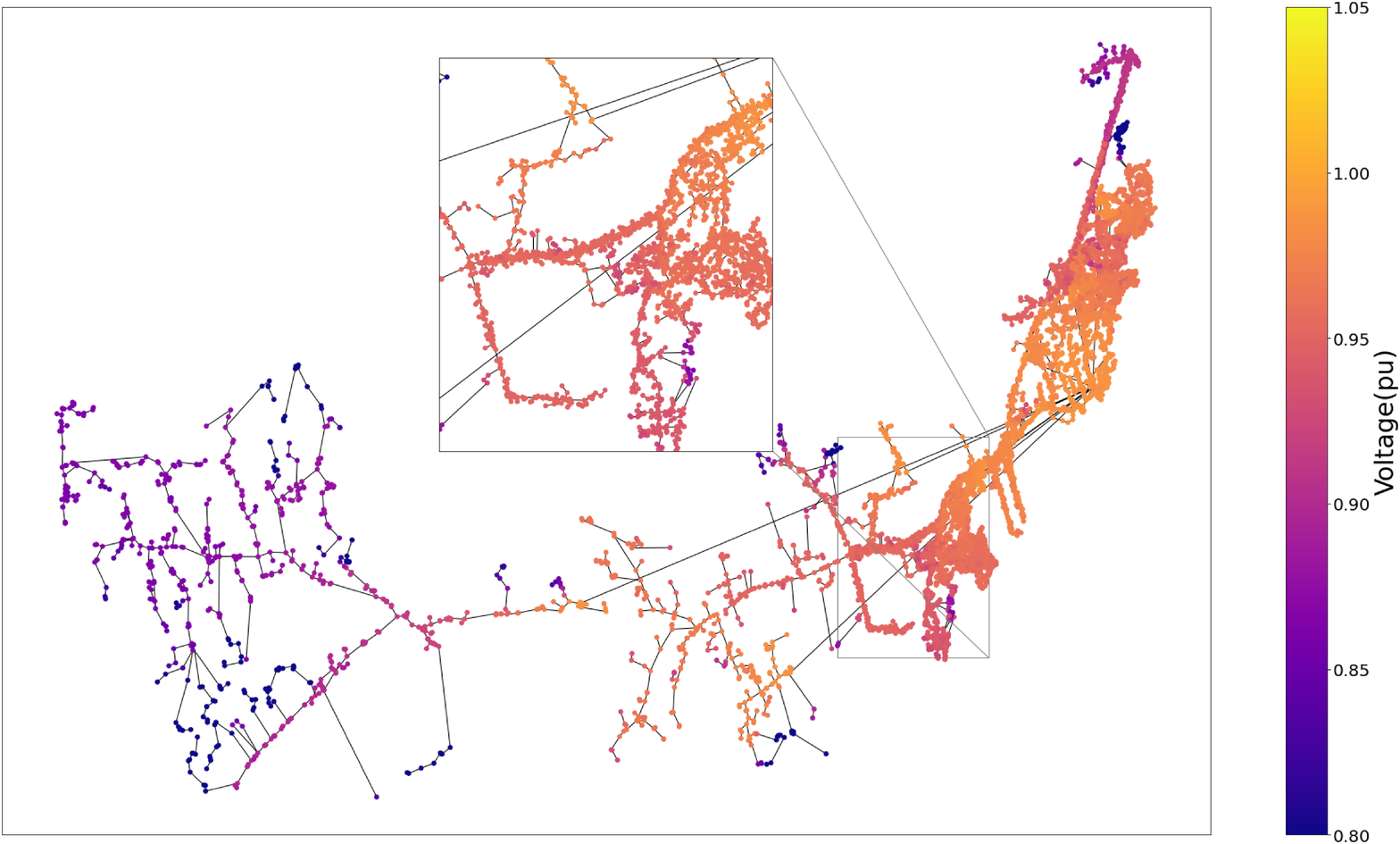}
	\includegraphics[width=0.24\textwidth]{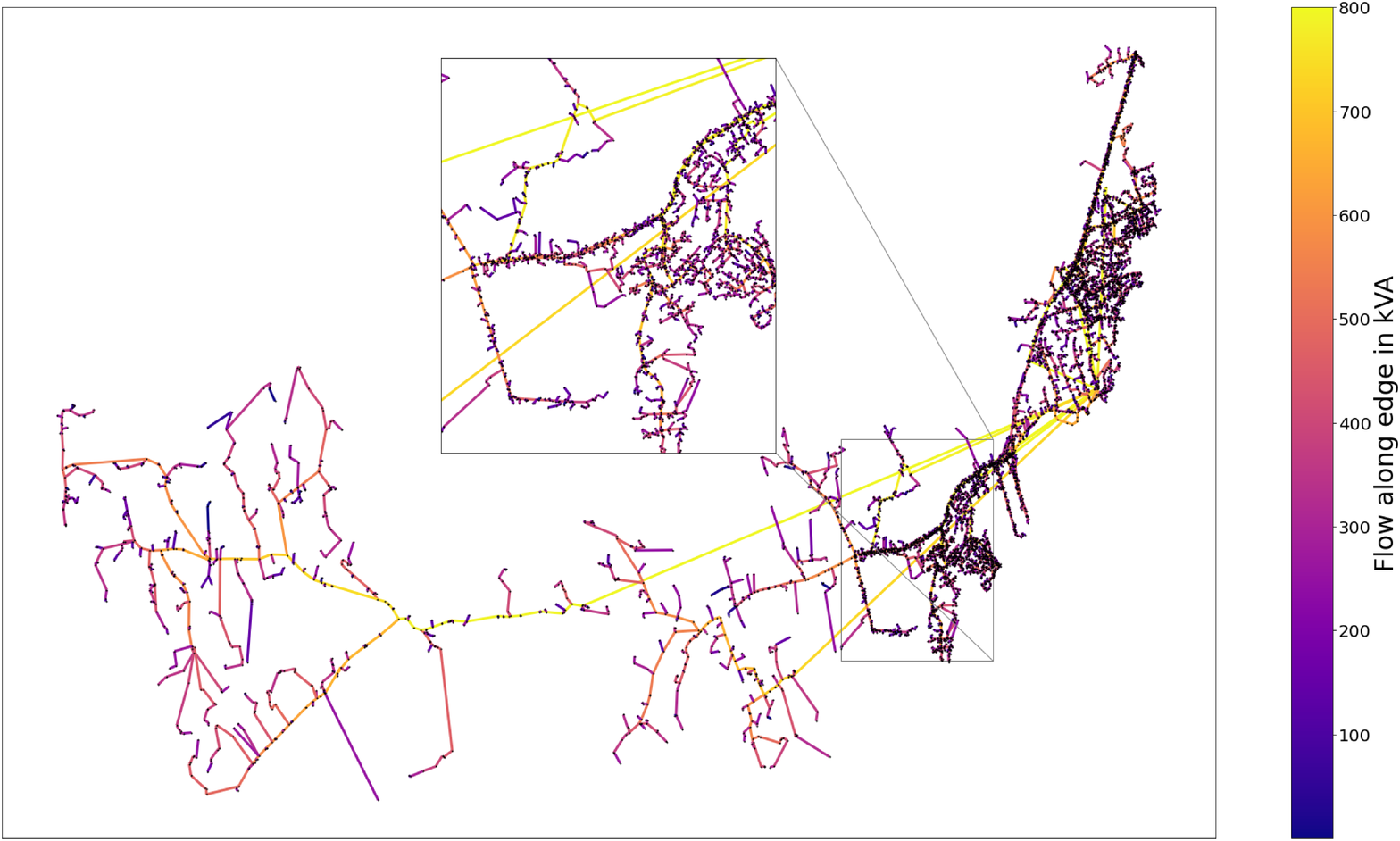}
	\caption{Node voltages and edge flows in the created synthetic network. The generated synthetic distribution network satisfies the usual power engineering operational constraints.}
	\label{fig:prim-validate}
\end{figure}
The synthetic network is generated following power engineering operational constraints such as the node voltages and edge flows. Considering that all residential customers are consuming average hourly load, a linearized distributed power flow program is run on the network. The operating voltage at each node is recorded and shown in the left figure of Fig~\ref{fig:prim-validate} through a colored map. The ANSI C84.1 standard mandates a $\pm 5\%$ voltage regulation for normal operating condition. It is observed that the voltage profile is within acceptable limits of 0.95 to 1.0 pu. The voltage profile at the extreme ends can be further improved by optimally placing capacitors along the network. The loading level of each line is also depicted in the right figure of Fig~\ref{fig:prim-validate}. The loading level of lines near the residences are small where as those for the feeders are high, which is expected to occur for a radial distribution network. This completes the operational validation of the network where we can conclude that when all residential customers consume average hourly load, the power distribution network operates in a \emph{secure} state with voltages and line flows within acceptable operational limits.

\begin{table}[htbp]
	\centering
	\caption{Residence and road network statistics}
	\label{tab:data-stat}
	\begin{tabular}{ccllll}
		\hline
		\multicolumn{2}{l}{} & \textbf{Mean} & \textbf{Median} & \textbf{Minimum} & \textbf{Maximum} \\ \hline
		\multicolumn{2}{c}{Residences} & 24907 & 11454 & 1940 & 404109 \\ \hline
		\multirow{2}{*}{\begin{tabular}[c]{@{}c@{}}Road\\ Network\end{tabular}} & Nodes & 6320 & 4090 & 505 & 71126 \\ \cline{2-6} 
		& Edges & 7438 & 4675 & 610 & 86514 \\ \hline
	\end{tabular}
\end{table}
\noindent\textbf{Computational scalability.}~We generate synthetic distribution network for the entire state of Virginia by solving two MILP optimization problems. There are $134$ counties and independent cities and residence and road network information is extracted for each of them individually. The county-wise statistics of this extracted data is provided in Table~\ref{tab:data-stat}. In this section, we discuss about the factors affecting the computational speed of the generation procedure.
\begin{table}[htbp]
	\centering
	\caption{Summary of time complexity}
	\label{tab:time}
	\begin{tabular}{|l|l|l|l|l|}
		\hline
		\textbf{\begin{tabular}[c]{@{}l@{}}Algorithm \\ Steps\end{tabular}} & \multicolumn{1}{c|}{\textbf{\begin{tabular}[c]{@{}c@{}}Step 1\\ Mapping\end{tabular}}} & \multicolumn{1}{c|}{\textbf{\begin{tabular}[c]{@{}c@{}}Step 2\\ Secondary\\ Network\end{tabular}}} & \multicolumn{1}{c|}{\textbf{\begin{tabular}[c]{@{}c@{}}Intermediate \\ step\\ Voronoi \\ partitioning\end{tabular}}} & \multicolumn{1}{c|}{\textbf{\begin{tabular}[c]{@{}c@{}}Step 3\\ Primary\\ Network\end{tabular}}} \\ \hline
		\textbf{\begin{tabular}[c]{@{}l@{}}Worst time \\ complexity\end{tabular}} & $\mathcal{O}(n)$ & $\mathcal{O}(2^n)$ & $\mathcal{O}(n\log n)$ & $\mathcal{O}(2^n)$ \\ \hline
		\textbf{\begin{tabular}[c]{@{}l@{}}Maximum\\ computation \\ time\end{tabular}} & 5 mins & 23 hours & 3 mins & 28 hours \\ \hline
	\end{tabular}
\end{table}

All the experiments (county level secondary network generation, primary network generation for Voronoi partition) are conducted on a node with CentOS 7.6.1810 x8664 having 40  cores - each core running Intel Xeon Gold 6148 CPU@ 2.40 GHz. The total available memory on every node is 150GB, shared by the 40 cores on the node. Table~\ref{tab:time} summarises the worst time complexity and worst computation time for each step in the algorithm. 

For the secondary network generation, we observe that the time to create the connection between the residences and local transformers increases exponentially with the number of residences to be connected. The box plot shows significant level of variation for same range of residences. This is due to the pattern in which they are distributed around the road network link. A congested distribution of residences at particular spots along a road link requires comparably more time than the cases where residences are uniformly placed along the link. For the primary network generation, the computation time also exponentially increases with the number of binary variables involved in the problem.  
\begin{figure}[htbp]
	\centering
	\includegraphics[width=0.24\textwidth]{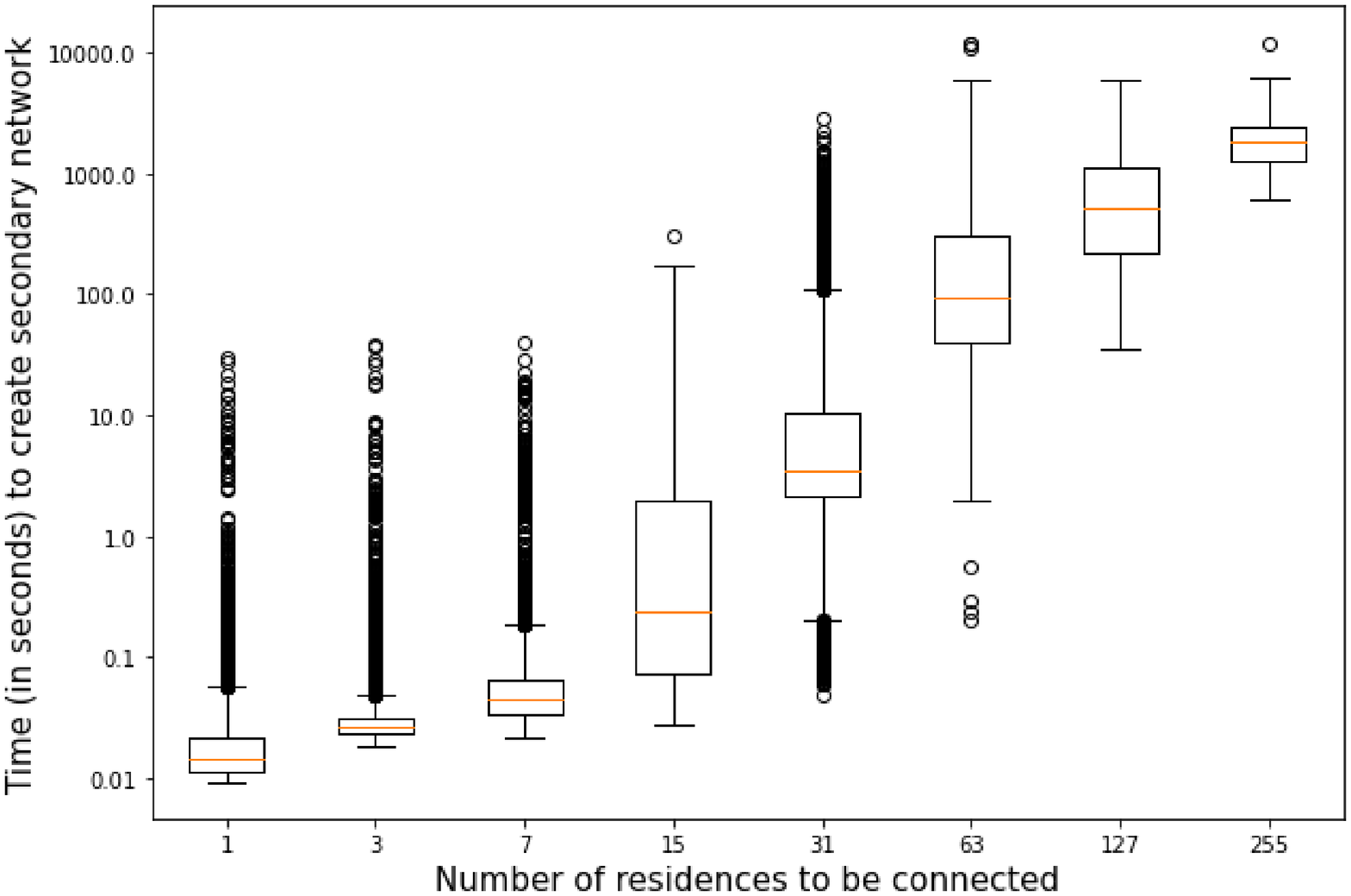}
	\includegraphics[width=0.24\textwidth]{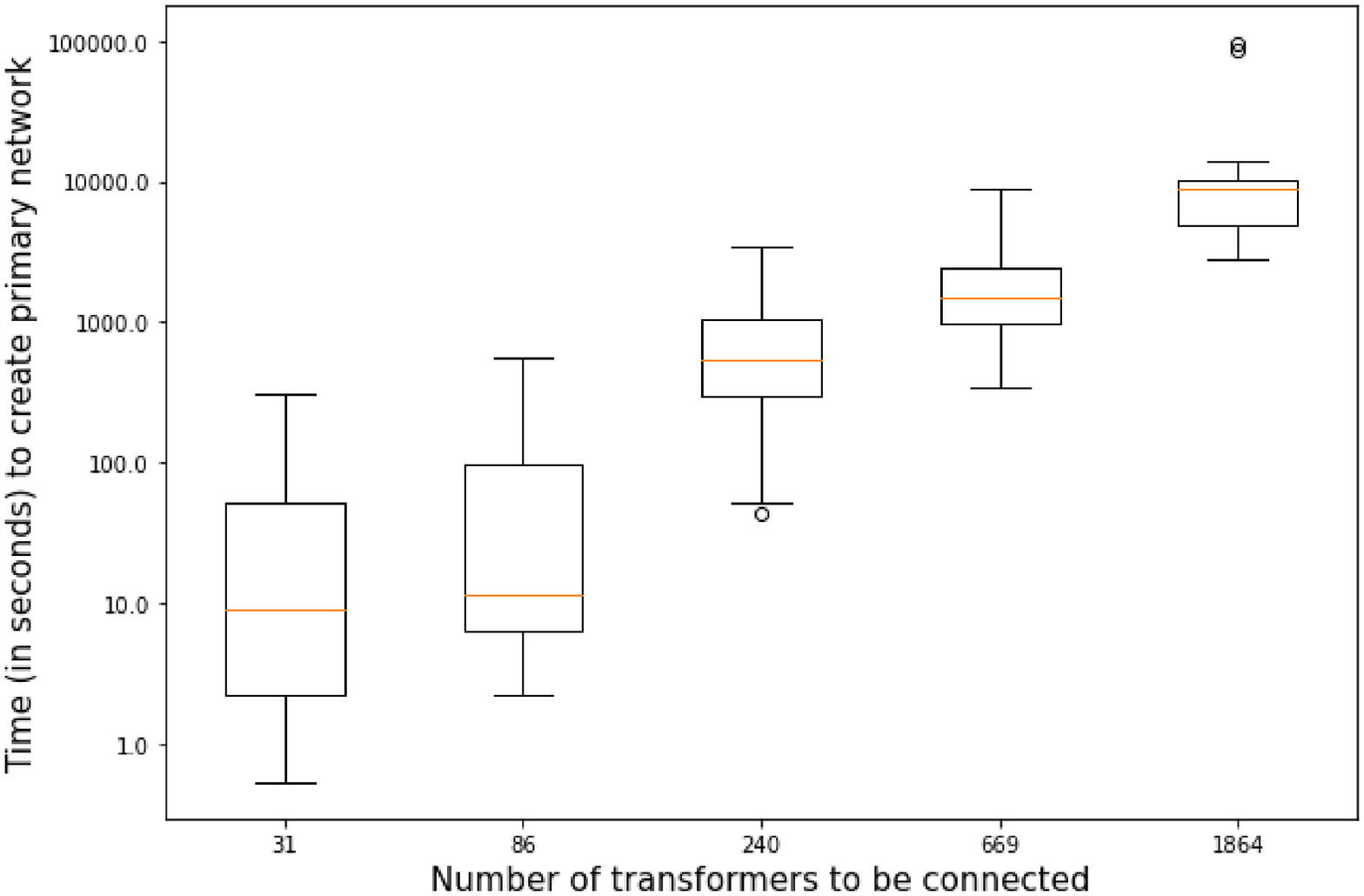}
	\caption{Box plots showing time required to create the secondary (left) and primary (right) distribution networks. The time required increases exponentially with the number of binary variables. The number of binary variables is proportional to number of residences (for secondary network) and transformer nodes (for primary network) to be connected.}
	\label{fig:secondary-time}
\end{figure}

\section{Validation and Verification}\label{sec:valid}
The primary goal of the paper is to generate realistic synthetic distribution networks which resemble real-world networks. We obtained real world power distribution networks for town of Blacksburg in south-west Virginia. However, the actual network has been incrementally built over a long period of time. Therefore, it is expected that there would be structural differences between the two networks. In this section, we attempt to validate the proposed methodology for generating networks

\noindent\textbf{Variation in network structure.} Fig.~\ref{fig:validate-structure} shows the structural comparison between the actual distribution network of a region and the synthetic network generated for the same. A zoomed in view (bottom right) of a portion of the networks shows that they are similar to each other. While the synthetic network almost retraces the road network, the actual network is adjacent to it. This validates our primary assumption that the primary distribution network follows the road network. The actual network has multiple redundancies in its structure which is present to increase the resiliency of the distribution network. Our work does not integrate such attributes while generating the synthetic networks which results in a different structure from the real network.

\noindent\textbf{Feeder selection.} The significant structural difference between the two networks is the substation feeder to which it is connected. The actual network is connected to the nearest geographically located substation, where as the synthetic network is connected to a distant feeder. It is to be noted that the synthetic distribution network is generated using first principles with the assumption that it follows the road network to the maximum extent (it is easier to place distribution poles along the road network). The zoomed in view on the center right of Fig~\ref{fig:validate-structure} provides the explanation for the structural difference in the two networks. It is seen that the nearest substation (from which the actual network is fed) has no road links connecting itself to the neighborhood. On the contrary the other substation, though located at a further distance geographically, has road links connecting itself to the neighborhood under consideration. The intermediate step in our algorithm assigns nodes to substations through a shortest network distance based Voronoi partitioning. Therefore, the synthetic network generation algorithm assigns the nodes in the neighborhood to the latter substation.
\begin{figure}[htbp]
	\centering
	\includegraphics[width=0.485\textwidth]{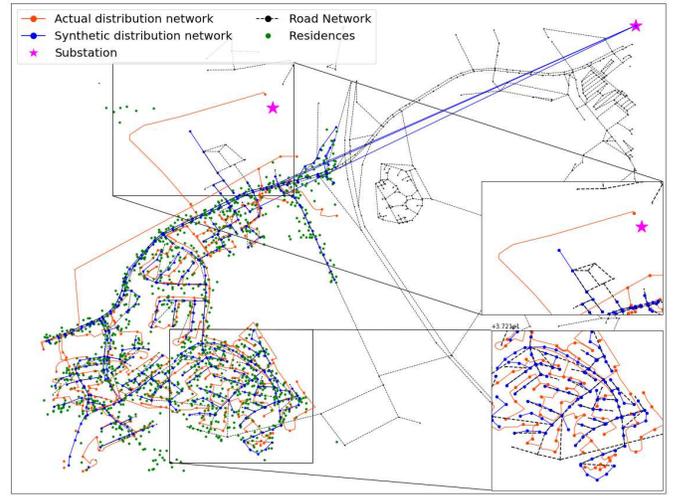}
	\caption{Comparing the structure of actual distribution network (red) with the synthetic distribution network (blue) generated using first principles from the knowledge of road network (black). Inset figure on bottom right shows structural similarities between the pair of networks. The inset figure on center right shows the absence of road links near the geographically nearest substation which feeds the actual network. Our synthetic network is fed from the substation which is closest along the road network but is geographically further.}
	\label{fig:validate-structure}
\end{figure}

Note that in the proposed algorithm, the substation feeders need not have equal sharing of load which is similar to the real world  situation wherein substations have different capacities. The maximum capacity of each feeder (or feeder rating) can be pre-defined or randomly selected and used as a constraint in the network generation algorithm.

\begin{figure}[htbp]
	\centering
	\includegraphics[width=0.24\textwidth]{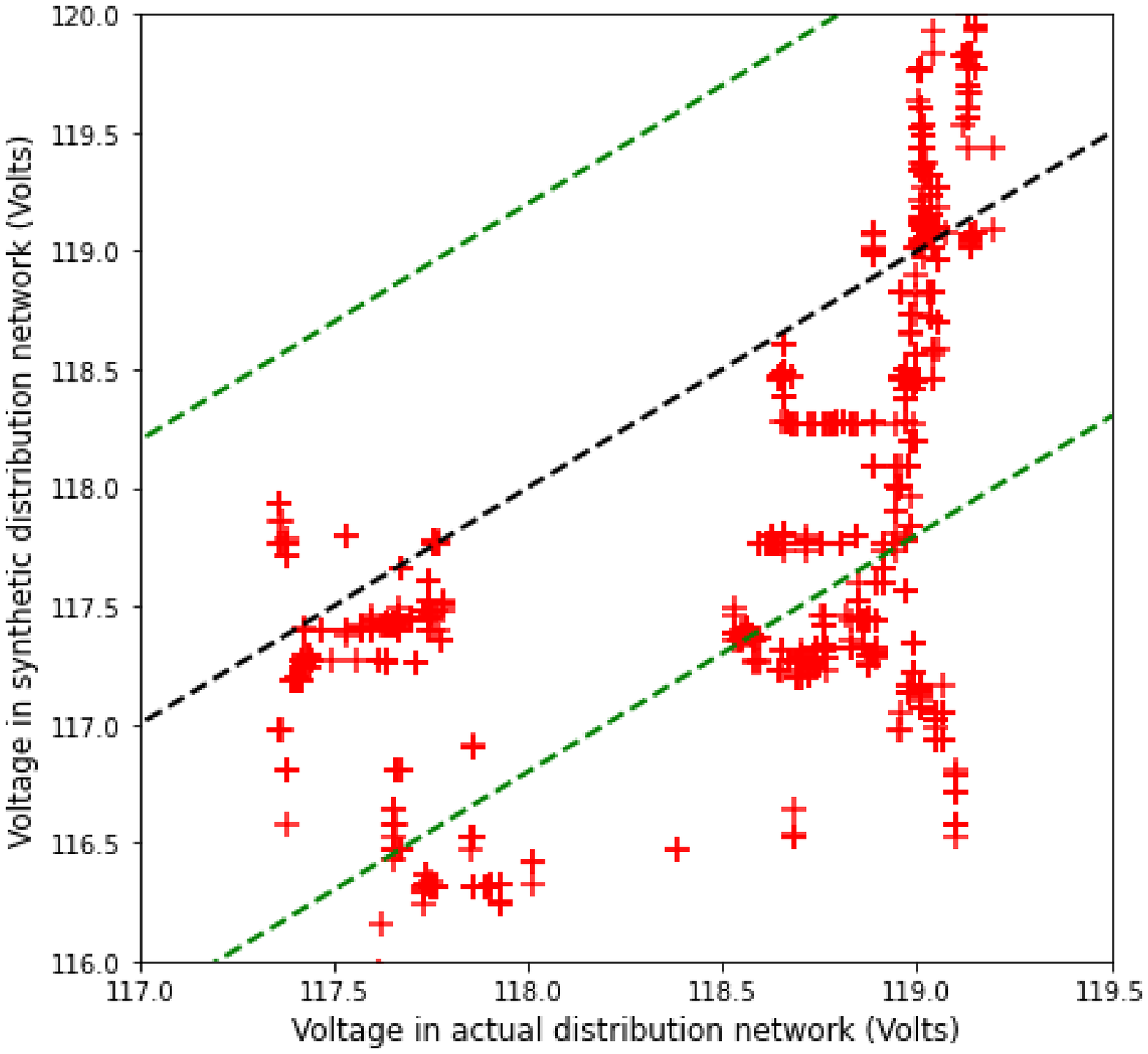}
	\includegraphics[width=0.24\textwidth]{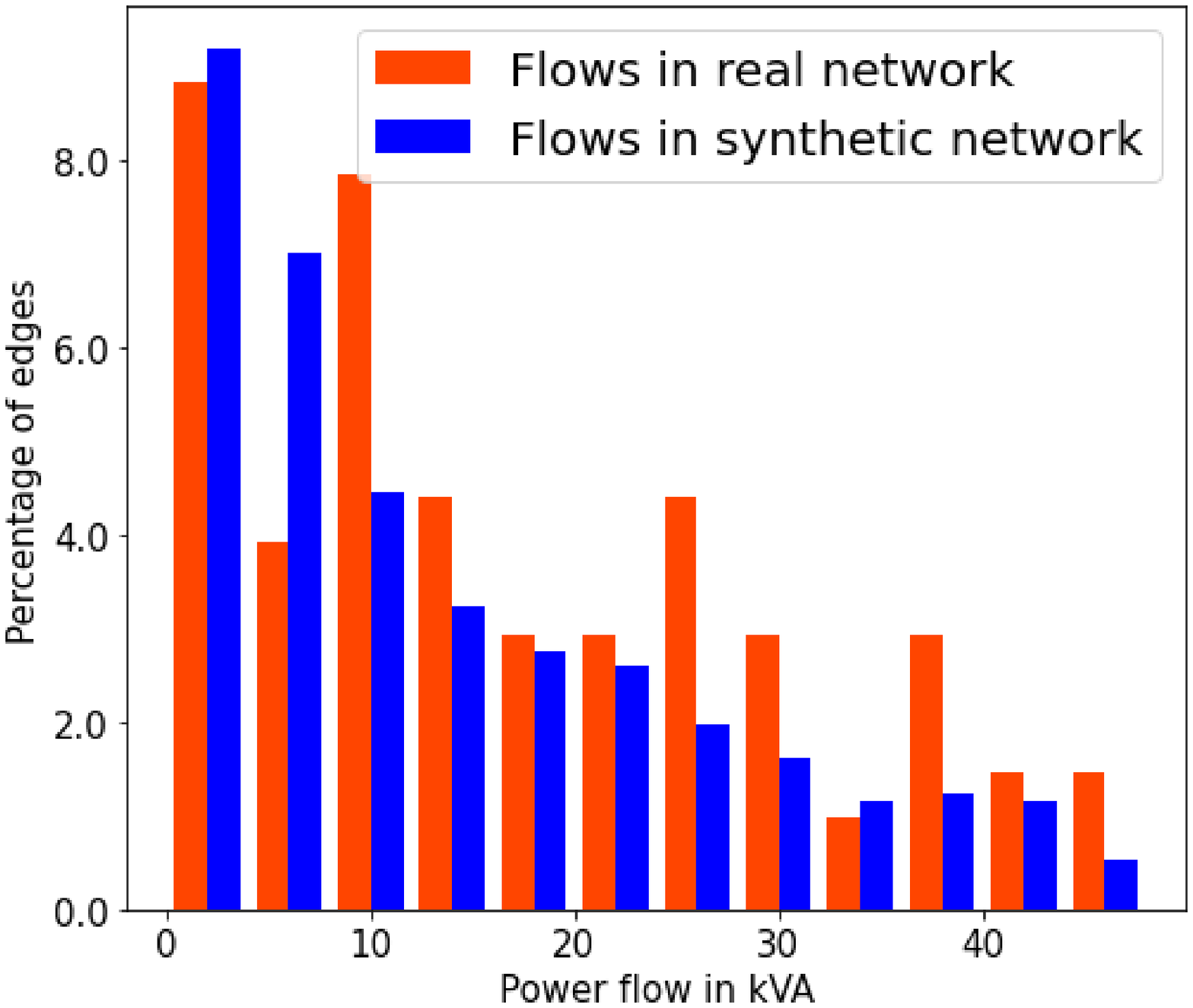}
	\caption{Comparing the residential node voltages and edge power flows for the two networks. The voltages in the synthetic network are within $\pm1\%$ voltage regulation of the voltages in the actual network. The observable differences in the edge flow distribution of the two networks are primarily due to their structural dissimilarities.}
	\label{fig:validate-voltage-flow}
\end{figure}
\noindent\textbf{Variation in node voltage profile and edge flows.} From the above discussion, it is evident that though the actual and synthetic network are similar for most part,  the feeders connecting the network to the substation are often different which leads to the observed differences. We now compare the voltages at the residences when they are connected to the actual and synthetic network in Fig.~\ref{fig:validate-voltage-flow}. We call this \emph{operational validity}: The basic idea is that if we substitute the actual network with the synthetic network, we should see minimal node voltage differences between the two. Here, the green lines signify $\pm1\%$ deviation from the voltages if the residences had been connected to the actual network. We observe that the residence voltages in the synthetic network remain within this $\pm1\%$ regulation, which validates our generated synthetic network. We also compare the edge flows in the two networks through the histogram in Fig.~\ref{fig:validate-voltage-flow}. Since the networks are radial in structure, the edge flows is proportional to the number of children nodes. The structural differences in the two networks explain the observed deviation in the edge flows of the networks. Note that our methodology aims to provides a generic framework to create multiple synthetic distribution networks which resemble actual networks but do not exactly mimic them. One can fine-tune the parameters/constraints in the optimization problem to obtain a different structure.

\section{Conclusion}
We present a methodology to generate synthetic distribution networks for a given geographic location based on the available road network data, which maintains a radial structure and minimizes the total length. We have generated synthetic distribution network for the state of Virginia and proposed methodologies to scale our optimization algorithms for regions of large size. Such realistic synthetic networks are useful in a number of applications when studying power systems engineering. Finally we provided a structural and operational validation of the generated networks by comparing them with real distribution networks.

The consideration of economic, engineering, behavioral and geographical aspects for inferring the distribution network has allowed us to choose certain parameters (e.g., feeder capacities, maximum number of feeders etc.)---these parameters can be pre-defined by user's requirement or can be randomly selected to create ensembles of networks, which can be a direction for further research. Another important aspect of real distribution networks is the existence of multiple redundancies to increase the resiliency of the distribution system, which is a future research direction.

\noindent
\textbf{Acknowledgements.}
This work was partially supported by the NIH Grant 1R01GM109718; NSF Grants: BIG DATA IIS-1633028, OAC-1916805, Expeditions in Computing CCF-1918656, RAPID OAC-2027541, CMMI-1832587; NASA Applied Sciences Program Grant \#80NSSC18K1594; and DARPA under Contract No. HR0011-19-C-0096.

\bibliographystyle{IEEEtran}
\bibliography{references}

\end{document}